\documentclass{IEEEtran}

 \usepackage{pstricks-add,pst-node,pst-tree}
 \usepackage{epsfig}
 \usepackage{pst-grad} 
 \usepackage{pst-plot} 
 \usepackage[space]{grffile} 
 \usepackage{etoolbox} 
 \makeatletter 
 \patchcmd\Gread@eps{\@inputcheck#1 }{\@inputcheck"#1"\relax}{}{}
 \makeatother

\usepackage{mathdots}
\usepackage[english]{babel}
\usepackage{relsize}
\usepackage{graphics,graphicx}
\usepackage{url}

\usepackage{booktabs}
\usepackage{multirow}
\usepackage{mparhack}
\usepackage{subfigure}
\usepackage{authblk}
\usepackage{amsthm}
\usepackage{mathrsfs}
\usepackage[noend]{algorithmic}
\usepackage[linesnumbered,ruled,vlined]{algorithm2e}

\usepackage{array}
\usepackage{cite}
\usepackage{eqparbox}
\usepackage{mdwmath}

\usepackage{epsfig}
\usepackage{xcolor}

\usepackage{mathtools,cuted}
\usepackage{bm}
\usepackage{bbold}
\usepackage[all]{xy}
\usepackage{etoolbox}

\usepackage[euler-digits,euler-hat-accent]{eulervm} 
                               
\DeclareMathAlphabet{\mathantt}{OT1}{antt}{li}{it}
\DeclareMathAlphabet{\mathpzc}{OT1}{pzc}{m}{it}

\DeclarePairedDelimiter\norm{\lVert}{\rVert}%

\usepackage{accents}
\usepackage{multicol}

\usepackage{lipsum,color}

\usepackage{tcolorbox}

\usepackage{cite}

\usepackage{pgfplots}
\usepackage{tikz}
\usetikzlibrary{plotmarks}
\usetikzlibrary{patterns}
\usetikzlibrary{shapes.geometric}
\usetikzlibrary{pgfplots.groupplots}

\newtheorem{theorem}{Theorem}

\newtheorem{lemma}[theorem]{Lemma}

\DeclareFontFamily{OT1}{pzc}{}
\DeclareFontShape{OT1}{pzc}{m}{it}%
  {<-> s * [1.1] pzcmi7t}{}
\DeclareMathAlphabet{\mathpzc}{OT1}{pzc}%
                     {m}{it}

\def\R{\mathcal{R}}
\def\I{\mathcal{I}}
\def\J{\mathcal{J}}

\def\S{\mathcal{S}}
\def\K{\mathcal{K}}

\renewcommand{\vec}[1]{\mathbold{#1}}
\renewcommand{\bm}[1]{\mathbold{#1}}

\let\oldnl\nl
\newcommand{\nonl}{\renewcommand{\nl}{\let\nl\oldnl}}

\title{User-centric Performance Optimization with Remote Radio Head Cooperation in C-RAN}

\begin{document}

\bstctlcite{IEEEexample:BSTcontrol}

\title{User-centric Performance Optimization with Remote Radio Head Cooperation in C-RAN}
\author{

    \IEEEauthorblockN{Lei~You and Di~Yuan}\\
    
    \IEEEauthorblockA{Department of Information Technology, Uppsala University, Sweden
    \\\{lei.you; di.yuan\}@it.uu.se
    }
    
    \thanks{The paper has been partly presented in 2017 International Workshop on Resource Allocation, Cooperation and Competition in Wireless Networks, WiOpt, Paris, May 2017~\cite{lei-wiopt}.}

}

\maketitle

\begin{abstract}
In a cloud radio access network (C-RAN), distributed remote radio heads (RRHs) are coordinated by baseband units (BBUs) in the cloud. The centralization of signal processing provides flexibility for coordinated multi-point transmission (CoMP) of RRHs to cooperatively serve user equipments (UEs). We target enhancing UEs' capacity performance, by jointly optimizing the selection of RRHs for serving UEs, i.e., CoMP selection, and resource allocation. We analyze the computational complexity of the problem. Next, we prove that under fixed CoMP selection, the optimal resource allocation amounts to solving a so-called iterated function. Towards user-centric network optimization, \textcolor[rgb]{0,0,0}{we propose an algorithm for the joint optimization problem, aiming at scaling up the capacity maximally} for any target UE group of interest. The proposed algorithm enables network-level performance evaluation for quality of experience. 
\end{abstract}

\section{Introduction}

\subsection{Background}

\IEEEPARstart{C}{loud} radio access network (C-RAN) enables virtualization of functionalities of base stations by centrally managing a ``cloud'' that is responsible for signal processing and coordination of geographically distributed remote radio heads (RRHs)~\cite{EricssonAB:2015tv}. The baseband units (BBUs) that are separately located in base stations under the traditional cellular network architecture, are centrally deployed in BBU pools in the cloud in the C-RAN architecture. The centralization of signal processing enables coordination among RRHs. This facilitates the implementation of coordinated multipoint transmission (CoMP) for improving spectrum efficiency~\cite{Peng:2015bx}. The quality of service (QoS) may thus be enhanced by squeezing more out of the spectrum~\cite{EricssonAB:2015tv,Peng:2015bx}.

For the upcoming 5G, the concept of \textit{user-centric operation}~\cite{7411348,7880689,QoE-fairness,7462488} has been drawing attention recently. The paradigm of \textit{user-centric C-RAN} targets enhancing the quality of experience (QoE), which looks outward from the end-user. Whether or not the performance benefits from utilizing more resource is up to the type of service in use. For example, video and audio streaming, email and file transfers, or VoIP, all have different levels of sensitivities to network QoS metrics (e.g. throughput, delay, or packet loss) that are influenced by resource allocation. \textcolor[rgb]{0,0,0}{Compared to allocating network resource purely subject to the QoS fairness \cite{qos,QoE-def,QoE-fairness}, from the user-centric viewpoint, it is more rationale to allocate resource based on user groups of different service types}. In this context, \textit{optimizing the capacity for a specific target group of UEs} becomes relevant.

Resource allocation in user-centric C-RAN faces more challenges compared to the traditional long-term evolution (LTE) cellular networks. First, as CoMP is assumed to be put in use by default under the C-RAN architecture, the network becomes more connected, resulting in more complex coupling relations between the network elements. The network performance is also affected by selecting the serving RRHs for user equipments (UEs), i.e., CoMP selection. Given this background, the paper targets jointly optimizing CoMP selection and resource allocation in order to increase the capacity for \textit{any target group of users} in C-RAN. 

\subsection{Related Work}

\subsubsection{RRH selection and resource allocation}
A number of studies have focused on optimizing CoMP or resource allocation in C-RANs. 
In~\cite{Zhang:2015kr}, the authors studied CoMP-based interference
mitigation in heterogeneous C-RANs.
In~\cite{Davydov:2013cf}, the authors investigated the joint transmission (JT)
CoMP performance in C-RANs with large CoMP cluster size. The authors
of~\cite{Beylerian:2016do} investigated resource allocation of CoMP
transmission in C-RANs, and proposed a fairness-based scheme for
enhancing the network coverage. In~\cite{Ortin:2016jq}, the
authors studied the joint cell-selection and resource allocation problem, in
C-RANs without CoMP. In~\cite{Abdelnasser:2016ho}, a resource allocation
problem was studied for C-RANs with a framework of small cells
underlaying a macro cell\@. The study in~\cite{7543472} formulated an RRH selection optimization problem for power saving in C-RANs as a mixed integer linear programming model taking into account bandwidth allocation. A local search algorithm was proposed to solve the problem. In~\cite{7805409}, the authors jointly optimized RRH selection and power allocation to minimize the total transmit power of the RRHs. The file caching status in RRHs is part of the setup. The problem was formulated in a non-convex form and solved by a Lagrange dual method. The authors of~\cite{7938356} investigated the weighted sum rate problem by jointly optimizing RRH selection and power allocation. By applying a Lagrange dual method, the authors derived an optimal solution to a special case where the number of sub-carrier is infinite. 
In~\cite{7845615}, the authors studied the energy-aware utility maximization problem by jointly optimizing beamforming, BBU scheduling and RRH selection. The problem was decomposed and solved separately, yielding a heuristic solution. The authors of~\cite{7809637} formulated an energy minimization problem with joint optimization of resource allocation and RRH selection. A greedy strategy was employed for RRH activation and pairing active RRHs to UEs. Under fixed RRH-UE association, the corresponding resource allocation was computed. The authors of~\cite{7564824} employed a statistical model for characterizing the traffic load of RRHs in C-RAN. The load incurred by UEs is directly related to the number of allocated resource blocks (RBs). A heuristic dynamic RRH assignment algorithm was proposed. \textcolor[rgb]{0,0,0}{The authors in \cite{tatsis2019energy,6820391,ali2018coordinated} investigated the influence of RRH selection on the network resource usage efficiency and energy savings, while the dynamics and coupling between the resource consumption levels of RRHs are not taken into account.}

\subsubsection{QoS/QoE optimization in C-RAN}
In~\cite{Luo2016}, an abstract model of QoS-aware service-providing framework was proposed based on queueing theory. The model admits optimal solution obtained by convex optimization with respect to the service rate. In~\cite{7880689}, the authors investigated joint precoding and RRH selection subject to QoS requirement of users. The joint optimization problem was decoupled into two stages (resulting in suboptimality), respectively for RRH selection and power minimization. 
The authors of~\cite{DBLP:journals/corr/RenLPYH17} considered the single user case. A QoS-driven power- and rate-adaption scheme aiming at maximizing the user capacity was proposed. The authors showed the convexity of the formulation and solved it to the optimum with a Lagrange dual method. Reference~\cite{7405334} studied the problem of jointly optimizing RRH sleep control and transmission power over optical-fiber cables that connect RRHs to the cloud. The proposed network operation is based on a heuristic strategy aiming at minimizing the power consumption with guarantee of QoE. The QoE is defined to be the packet loss probability. In~\cite{7511364}, the authors proposed a beamforming scheme to coordinate RRHs for improving QoE. The QoE is defined to be the weighted sum of sigmoidal QoS functions of users. The problem is to maximize the QoE subject to power constraints. A heuristic algorithm was proposed.

\subsubsection{Resource allocation with load-coupling}
A related line of research is the characterization of resource allocation in orthogonal-frequency-division multiple access (OFDMA) networks. A model that characterizes the coupling relationship of allocated resource among network entities is becoming adopted~\cite{6732895,6204009,7880696,4212688,DBLP:journals/corr/LiaoAS16a,6853397,lei-wiopt}. By this model, a connection between network load and user bits demand in QoS/QoE satisfaction is established.

\subsection{Motivation}
One way of evaluating the performance with respect to QoS/QoE is to measure 
 \textit{how much of the user demand can be scaled up} before the network resource becomes exhausted.
\textcolor[rgb]{0,0,0}{For fairness-based QoS enhancement, the problem was studied by computing the \textit{maximum demand scaling factor} for all users \cite{DBLP:journals/corr/LiaoAS16a,6853397,lei-wiopt}, which is fundamentally based on computing the eigenvalue as well as the eigenvector for a non-linear system. References~\cite{DBLP:journals/corr/LiaoAS16a,6853397,lei-wiopt} employ a restricted setup of ours. Their solution approaches do not apply for our generalized scenario, as the resulting problem does not map to computing the eigenvalue and eigenvector anymore.} Whether or not the maximum capacity with respect to QoE-aware demand scaling can be effectively and efficiently computed remains open.
Besides, under the C-RAN architecture, the interplay between network resource allocation and RRH cooperation needs to be captured.
To the best of our knowledge, how to optimally perform \textit{user-centric demand scaling} has not been addressed yet.

\subsection{Contribution}

The main contributions of this paper are summarized as follows. \textcolor[rgb]{0,0,0}{We propose a new framework for computing the \textit{maximum 
demand scaling factor} for any given group of users in C-RAN. Our framework is a significant extension of the one used in~\cite{DBLP:journals/corr/LiaoAS16a,6853397,lei-wiopt}}. Furthermore, 
\textcolor[rgb]{0,0,0}{based on this framework, we study the joint optimization problem of time-frequency resource allocation and CoMP selection, in terms of user-centric demand scaling}. We address the tractability of this problem. To deal with the complexity, we propose an algorithm that alternates between CoMP selection and resource allocation. Specifically, \textcolor[rgb]{0,0,0}{we prove that, with fixed CoMP selection, the optimal resource allocation amounts to solving a so-called \textit{iterated function}}. Furthermore, we provide a partial optimality condition for improving CoMP selection and prove that it is naturally combined with our resource allocation method. The condition and the method are incorporated together to form our joint optimization algorithm. \textcolor[rgb]{0,0,0}{We remark that the solution method for demand scaling in \cite{DBLP:journals/corr/LiaoAS16a,6853397,lei-wiopt} can only address a special case of ours. We further proved that, under fixed CoMP selection, our proposed method solves the resource allocation to global optimality.}
Finally, we show numerically how the joint optimization scheme can be used to scale up user demands for user-centric capacity enhancement. The obtained results reveal how CoMP improves the user capacity and how the user group size and the number of CoMP users affect the performance. \textcolor[rgb]{0,0,0}{To the best of our knowledge, this is the first paper that addresses user-centric demand scaling with load-coupling and CoMP.}

\subsection{Paper Organization}
The paper is organized as follows. Section~\ref{sec:sys_model} gives the system model. Section~\ref{sec:computation} formulates the problem and analyzes its tractability. Section~\ref{sec:solving} derives our solution method for problem solving. After discussing the numerical results in Section~\ref{sec:numerical}, the paper is concluded in Section~\ref{sec:conclusion}. The proofs of all theorems in Section~\ref{sec:solving} are detailed in the Appendix. Throughout all sections, we use bold fonts to represent vectors/matrices, and capitalized letters in calligraphy to represent mathematical sets. As for function/mapping definitions, we use the notation $g: \mathtt{var}\mapsto \mathtt{expr}$ to represent a function/mapping $g(\mathtt{var})$ with mathematical expression $\mathtt{expr}$ of variable $\mathtt{var}$.

\textcolor[rgb]{0,0,0}{We use $\mathbb{R}_{++}$ to refer to all positive real numbers, i.e. $\mathbb{R}_{++}=(0,\infty)$. Similarly, we use $\mathbb{R}_{+}$ to refer to all non-negative real numbers, i.e. $\mathbb{R}_{+}=[0,\infty)$.}

\section{System Model}
\label{sec:sys_model}

\subsection{Notations}
Denote by $\R=\{1,2,\ldots,m\}$ the set of RRHs in a C-RAN\@. Denote by
$\J=\{1,2,\ldots,n\}$ the set of UEs. 
We use matrix
$\bm{\kappa}\in{\{0,1\}}^{m\times n}$ to indicate the association
between RRHs and UEs. The matrix $\bm{\kappa}$ is subject to optimization. For the sake of presentation, let us consider any given $\bm{\kappa}$.
For this $\bm{\kappa}$, we use  
$\R_j$ and $\J_i$ as generic notations for the set of RRHs serving UE $j$ and the set of UEs served by RRH $i$, respectively.
 Specifically, $i\in\R_j$ and $j\in\J_i$ if and only if $\kappa_{ij}=1$. Note that $\bm{\kappa}$ characterizes CoMP selections, i.e., which UE is served by which RRHs. Single antenna case in C-RAN \cite{FeAnEr19, FaAiAm19, AlAhKh19} is considered in this paper.

\subsection{CoMP Transmission}

Downlink is considered. Denote by $p_i$ the transmit power of RRH $i$ ($i\in\R$) on one resource block (RB). Denote by $h_{ij}$
the channel gain between RRH $i$ and UE $j$. Let $x$ be the channel input symbol sent to UE $j$ by the RRHs in 
$\R_j$. Entity $x_k$ denotes the channel input symbol sent by the other
RRHs that are not cooperatively serving UE $j$. The received channel output at
UE $j$ can be written as 
\begin{equation}
s=\sum_{i\in\R_j}\sqrt{p}_{i}h_{ij}x +
\sum_{k\in\R\backslash\R_j}\sqrt{p}_{k}h_{kj}x_{k} + \sigma_j.
\end{equation}
We assume that ${x}$ and ${x}_k$ ($k\in\R\backslash\R_j$) are independent
zero-mean random variables of unit variance. Parameter $\sigma_j$ models the noise of any user $j$. The signal-to-interference and noise ratio (SINR) of UE $j$ is given below, by 
~\cite{Nigam:2014cd},
\begin{equation}
    \gamma_j =
    \frac{|\sum_{i\in\R_j}\sqrt{p}_{i}h_{ij}|^2}{\sum_{k\in\R\backslash\R_j}
    w_{kj}+\sigma^2_j},
\end{equation}
where $w_{kj}$ is the interference received at UE $j$ from RRH $k$. \textcolor[rgb]{0,0,0}{For any user, the interference it receives comes from those RRHs that are not serving the user, and the cooperative RRHs do not generate interference to the user \cite{Nigam:2014cd}.}
\textcolor[rgb]{0,0,0}{We remark that the gains and noises can be different for users, and our solution method proposed later enables performance evaluation for scenarios with different user channel gains and noises.}

\subsection{Interference Modeling}

We introduce the network-level interference model that is widely adopted in OFDMA systems, referred to as ``\textit{load-coupling}''\cite{6732895,6204009,7880696,4212688,DBLP:journals/corr/LiaoAS16a,6853397,lei-wiopt}. The model is shown to capture well the characterization of interference coupling~\cite{6732895}. 
Denote by $M$ the total number of RBs in consideration. Define $\rho_{k}$ as the proportion of allocated RBs of RRH $k$ to serve all of its UEs. If $\rho_{k}=0$, it means that there is no time-frequency resource in use by RRH $k$. In this case, $k$ does not generate interference to others and $w_{kj}=0$ ($j\in\J\backslash\J_k$). On the other hand, if $\rho_{k}=1$, then all the RBs in RRH $k$ are used for transmission and RRH $k$ constantly interferes with others, i.e., $w_{kj}=p_{k}|h_{kj}|^2$ ($j\in\J\backslash\J_k$) \cite{Nigam:2014cd}. For $0<\rho_k<1$, $\rho_{k}$ serves as a scale parameter of interference:
\begin{equation}
w_{kj}=p_{k}|h_{kj}|^2\rho_{k}.
\label{eq:w}
\end{equation}

By the definition of $\rho_{k}$, it is referred to as \textit{load of the RRH} $k$ and can be intuitively explained as the likelihood that RRH $k$ interferes with others.
The network load vector is represented as $\bm{\rho}=[\rho_1,\rho_2,\ldots,\rho_m]$. Increasing any load $\rho_{k}$ may lead to the capacity enhancement of UEs in $\J_k$. On the other hand, as can be seen from~\eqref{eq:w}, the increase of $\rho_{k}$ results in higher interference from $k$ to other UEs $\J\backslash\J_k$, which may cause the load levels of RRHs other than $k$ to increase. Note that a heavily loaded RRH interferes more severely to others, while an RRH that is slightly loaded tends not to generate much interference. 

\subsection{User Demand Scaling}
\label{subsec:sys_mod-scaling}
Denote by $B$ the bandwidth per RB. The achievable bit rate for UE $j$ by the transmissions of $j$'s serving RRHs is denoted by a function $C_j:\mathbb{R}_{+}^{n}\rightarrow\mathbb{R}_{+}$ of SINR, which in turn is a function of the network load $\bm{\rho}$:
\begin{equation}
C_j: \bm{\rho}\mapsto MB\log_2(1+\gamma_j(\bm{\rho})).
\label{eq:cj} 
\end{equation}
\textcolor[rgb]{0,0,0}{Denote by $d_j$ the bits demand of UE $j$ ($j\in\J$). Given the proportion of allocated RBs in RRHs that are not serving UE $j$, the expression $\frac{d_j}{C_j(\bm{\rho})}$ gives the proportion of required RBs for the RRHs serving UE $j$ to satisfy this demand $d_j$. We remark that $C_j(\bm{\rho})$ is non-linear of the proportion of allocated RBs in the interfering RRHs. } Let $\mu_j$ represent the proportion of RBs allocated to $j$ by $j$'s serving RRHs. As for all UEs in $\J$, we let $\bm{\mu}=[\mu_1,\mu_2,\ldots,\mu_n]$. Because of CoMP, the RBs used by all these RRHs for serving $j$ are the same. For allocating sufficient proportion of RBs to satisfy UE $j$'s demand, we have:
\begin{equation}
\mu_j \geq \frac{d_j}{C_j(\bm{\rho})}.	
\label{eq:mu}
\end{equation}
By the definition of $\rho_{i}$ ($i\in\R$), we have 
\begin{equation}
\rho_{i}=\sum_{j\in\J_i}\mu_{j}. 
\label{eq:rhoi}
\end{equation}
 Denote by $\bar{\rho}$ the load limit of RRHs. Then we need to keep $\rho_i\leq\bar{\rho}$ ($i\in\R$), otherwise the network is overloaded, meaning that the available resource is not sufficient for delivering the demands. 
Combining~\eqref{eq:cj} with~\eqref{eq:rhoi}, for any $j\in\J$,  we use $f_j:\mathbb{R}_{+}^{n}\mapsto\mathbb{R}_{++}$ to denote the following function:
\begin{equation}
f_j:\bm{\mu}\mapsto \frac{d_j}{C_j(\bm{\mu})}.	
\label{eq:f_j}
\end{equation}

Given $\vec{d}=[d_1,d_2,\ldots,d_n]$, consider for QoE of a specific
UE $j$, we would like to scale up $d_j$ by a \textit{demand scaling
factor} $\alpha$ ($\alpha>0$). The reason of considering the 
maximization of $\alpha$ is it tells how much traffic growth the network
can accommodate by optimizing user association, as such
it provides information related to network capacity.
With demand scaling, $\mu_j$ needs to satisfy:

\begin{equation}
\mu_j\geq  \alpha f_{j}(\bm{\mu}).
\label{eq:muj}
\end{equation}
Due to the mutual coupling relationship of the elements in the vector $\bm{\mu}$, scaling up the demand for UE $j$ may cause the increase of the  interference to others such that the bits demand of other UEs cannot be satisfied. Therefore, one should also make sure that the following equation holds:
\begin{equation}
\bm{\mu}_{-j}\geq \vec{f}_{-j}(\bm{\mu}).
\label{eq:mu-j}
\end{equation}

where $\bm{\mu}_{-j}$ and $\vec{f}_{-j}$ ($j\in\J$) represent the
vectors without the $j_{\text{th}}$ element.  Finally, the resource
limits are subject to the inequalities $\sum_{j\in\J_i}\mu_j\leq
\bar{\rho}$ ($i\in\R$).  More generally, one can scale up the demand
for any UE group $\S$ ($\S\subseteq\J$). One way to view $\S$ is the
set of users with elastic traffic.  For the other users in $\J
\setminus \S$, their demand has to be satisfied without any scaling.
For the special case of $\S=\J$, maximum demand scaling corresponds to
finding out how much traffic growth is possible before the resource is
exhausted.

Note that~\eqref{eq:muj} and~\eqref{eq:mu-j} form a system of
non-linear inequalities in terms of $\bm{\mu}$ and $\alpha$, which
cannot be readily solved. There is a special case that is however
easy. \textcolor[rgb]{0,0,0}{With $\S=\J$, one can
write~\eqref{eq:muj} and~\eqref{eq:mu-j} as
$\bm{\mu}\geq\alpha\vec{f}(\bm{\mu})$, i.e., the demand scaling is on
all the UEs.}  To use the minimum amount of resource to satisfy the
demands, we have $\frac{1}{\alpha}\bm{\mu}=\vec{f}(\bm{\mu})$. In this
case, $\frac{1}{\alpha}$ and $\bm{\mu}$ are exactly the eigenvalue and
eigenvector of this non-linear equations system and can be solved by
the concave Perron-Frobenius Theorem~\cite{Krause:2001wd}. However,
the conclusion does not hold for the general case
$\S\subseteq\J$. Indeed, the general case is fundamentally different,
as $\alpha$ is not the scaling parameter in each dimension of the
function $\vec{f}$ when $\S\subset\J$.

\section{Problem Formulation and Complexity Analysis}
\label{sec:computation}

In this section, we formulate the demand scaling problem and prove its computational complexity. 
Recall that $\R_j$ and $\J_i$ characterize the CoMP selection and are induced by the matrix $\bm{\kappa}$.
To characterize CoMP selection, 
we treat $\R_j$ ($j\in\J$) and $\J_i$ ($i\in\R$) as mappings shown below, which map any matrix $\bm{\kappa}$ to the sets of serving RRHs and served UEs, respectively.
\begin{equation}
\R_j:\bm{\kappa}\mapsto \{i|\kappa_{ij}=1\}	
\end{equation}
\begin{equation}
\J_i:\bm{\kappa}\mapsto	\{j|\kappa_{ij}=1\}
\end{equation}
The problem is formulated in~\eqref{eq:maxd}. \textcolor[rgb]{0,0,0}{Solving \eqref{eq:maxd} is equivalent to solving the maximization of throughput satisfaction ratio for $\S$ subject to the max-min fairness\footnote{\textcolor[rgb]{0,0,0}{We remark that \eqref{eq:jinS} can be equivalently re-written as $\alpha\leq \min_{j\in\S}{\mu_j C_j(\bm{\rho},\bm{\kappa})}/{d_j}$. Then \eqref{eq:maxd} can be reformulated as 
\[
\max\limits_{\bm{\mu}\geq\bm{0},\bm{\kappa}}~\min_{j\in\S} \frac{\mu_j C_j(\bm{\rho},\bm{\kappa})}{d_j} \text{ s.t. \eqref{eq:jnotinS}--\eqref{eq:kappaij} } 
\]
which is exactly the maximization problem for user-specific throughput satisfaction ratio. The demand $d_j$ is satisfiable if and only if this ratio is greater than or equal to one.}} \cite{DBLP:journals/corr/LiaoAS16a}.}
\begin{subequations}
\begin{alignat}{2}
[\textit{MaxD}]\quad & 
    \max\limits_{\alpha> 0,\bm{\mu}\geq\bm{0},\bm{\kappa}}~\alpha  \label{eq:haha}\\
    \textnormal{s.t.} \quad & \mu_j\geq\alpha f_j(\bm{\mu},\bm{\kappa}) \quad j \in\S \label{eq:jinS}\\
         \quad & \mu_j\geq f_j(\bm{\mu},\bm{\kappa}) \quad j \in \J\backslash\S \label{eq:jnotinS}\\
         \quad & \sum_{j\in\J_i}\mu_j\leq \bar{\rho} \quad i \in\R \label{eq:barrho}\\
         \quad & \kappa_{ij} \in \{0,1\} \quad i\in\R,~j\in\J \label{eq:kappaij}
\end{alignat}
\label{eq:maxd}
\end{subequations}

\textcolor[rgb]{0,0,0}{The objective is to maximize the demand scaling factor $\alpha$ for a given set of UEs $\S$ ($\S\subseteq\J$), given the base user demand $\bm{d}$ ($\bm{d}$ is in the expression of the function $f_j$). Remark that the function $f_j(\bm{\mu})$, defined in~\eqref{eq:f_j}, computes the minimum required proportion of resource for satisfying the demand $d_j$ of UE $j$, and $\mu_j$ is the resource allocated to UE $j$. Therefore constraints~\eqref{eq:jnotinS} ensure that sufficient time-frequency resources are allocated for satisfying the demands $d_j$ ($j\in\J\backslash\S$). The UEs in $\S$ can be regarded as being throughput-oriented such that delivering more demands leads to higher satisfaction, as imposed by constraints~\eqref{eq:jinS}. Hence, the objective is maximizing the demand scaling factor $\alpha$ for $\S$.} Constraint~\eqref{eq:barrho} imposes the maximum RRH load limit. The variable matrix $\bm{\kappa}$ controls the CoMP selection.  

Solving \eqref{eq:maxd} yields an a network-level evaluation for user group data rate enhancement \cite{6732895,6204009,7880696,4212688,DBLP:journals/corr/LiaoAS16a,6853397}, with the channel coefficients and noises corresponding to approximate averages upon the time scale of interest of load coupling. This type of modeling is commonly used for studies and analysis at a macroscopic level, and in fact variations of gains and noises have little impact on the accuracy of the load coupling model~\cite{6732895}. In addition, we remark that the demand scaling factor is the satisfaction ratio of the user demands in $\S$. That is, given base demands $d_j$ ($j\in\J$) and the RRH resource limit $\bar{\rho}$, 
\textcolor[rgb]{0,0,0}{
the solution $\alpha^{*}$ obtained by solving \textit{MaxD} is indeed the maximum satisfiable ratio of $d_j$ ($j\in\S$) within the resource limit. Namely, if $d_j$ ($j\in\S$) are satisfiable, then we have $\alpha^{*}\geq 1$. Otherwise $\alpha^{*}<1$ holds and $\alpha^{*}$ is the satisfiable proportion of the demands $d_j$ ($j\in\S$). 
}

\textcolor[rgb]{0,0,0}{We remark that, the method of scaling the demand for $\J$, which is a special case of \textit{MaxD}, does not generalize to priority-aware resource optimization.}

Theorem~\ref{thm:NP-hard} below shows the $\mathcal{NP}$-hardness of \textit{MaxD}. \textcolor[rgb]{0,0,0}{The basic idea behind the proof is to show that there is one scenario that is at least as hard as \textit{3-SAT}.}
\begin{theorem}
\textit{MaxD} is $\mathcal{NP}$-hard.
\label{thm:NP-hard}
\end{theorem}
\begin{proof}
We prove the theorem by a polynomial-time reduction from the 3-satisfiability
(\textit{3-SAT}) problem that is $\mathcal{NP}$-complete. Consider a \textit{3-SAT} problem with
$N_1$ Boolean variables $b_1,b_2,\ldots,b_{N_1}$, and $N_2$ clauses. A
Boolean variable or its negation is referred to as a literal, e.g.  $\hat{b}_i$
is the negation of $b_i$. A clause is composed by a disjunction of exactly three
distinct literals, e.g.,  $(b_1\vee b_2 \vee\hat{b}_3)$.
The \textit{3-SAT} problem amounts to determining whether or not there exists an
assignment of true/false to the variables, such that all
clauses are satisfied. 

The corresponding feasibility problem of \textit{MaxD} is that whether or not there exists any $(\alpha,\bm{\mu},\bm{\kappa})$ such that constraints \eqref{eq:haha}--\eqref{eq:kappaij} are satisfied. 
To make the reduction, we construct a specific network scenario as follows. Suppose
we have $N_1+N_2+1$ UEs in total, denoted by
$v_{0},v_{1},v_{2},\ldots,v_{N_1+N_2}$, respectively. Also, we have in total
$2N_1+N_2+1$ RRHs, denoted by $a_1,\hat{a}_1,a_2,\hat{a}_2,\ldots,a_{N_1},\hat{a}_{N_1}$, and
$a_{0},a_{N_1+1},\ldots,a_{N_1+N_2}$. The RRHs $a_1,\hat{a}_1,a_2,\hat{a}_2,\ldots,a_{N_1},\hat{a}_{N_1}$ are the counter parts to the $2N_1$ variables and their negations. And the RRHs $a_{N_1+1},\ldots,a_{N_1+N_2}$ are corresponding to the $N_2$ clauses. For each $v_{i}$
($1\leq i\leq N_1$), we set $\{a_i,\hat{a}_i\}$ as its candidate RRHs. For $v_{0}$ and each $v_{j}$
($N_1<j\leq N_1+N_2)$, we have exactly one candidate RRHs. Therefore, their serving RRHs are fixed, i.e., $\R_{v_{0}}=\{a_0\}$ and $\R_{v_{j}}=\{a_j\}$. Let $p_{a_0}=3N_1+1$. For $1\leq i\leq N_1$, let $p_{a_{i}}=p_{\hat{a}_{i}}=3.0$. For $N_1<j\leq N_1+N_2$, let $p_{a_j}=3.0$.
 For UE $v_0$, $|h_{a_0,v_0}|^2=|h_{a_{i},v_0}|^2=|h_{\hat{a}_i,v_0}|=1.0$ ($1\leq i\leq N_1$). 
For any UE $v_i$ ($1\leq i\leq N_1$), $|h_{a_i,v_i}|^2=|h_{\hat{a}_i,v_i}|^2=1.0$.
For any UE $v_{j}$ ($N_1< j\leq N_1+N_2$), $|h_{a_i,v_{j}}|^2$ and $|h_{\hat{a}_i,v_{j}}|^2$  ($1\leq i\leq N_1$) equal $\frac{1}{3}$ if $b_i$ and $\hat{b}_i$ appears in clause $j$, respectively.
In addition, $|h_{a_{j},v_{j}}|^2=1.0$ ($N_1<j\leq
N_1+N_2$).
The gain values between all other RRH-UE
pairs are negligible, treated as zero. \textcolor[rgb]{0,0,0}{The noise power $\sigma^2_j$ ($\forall j\in\J$) is $1.0$.} We normalize the demands
of UEs by $B\times M$, such that $d_{v_i}=2.0$ ($1\leq i\leq N_1$) and $d_{v_0}=d_{v_j}=1.0$ ($N_1<j\leq N_1+N_2$).  
Below we establish connections between solutions of \textit{MaxD} and those of \textit{3-SAT}. For \textit{MaxD}, note that if $(1,\bm{\mu},\bm{\kappa})$ is not feasible, then $(\alpha,\bm{\mu},\bm{\kappa})$ with $\alpha>1$ is not either.

First, we note that each UE $j$ ($0\leq j\leq N_1+N_2+1$) should be served by at
least one RRH, otherwise $C_j$ equals $0$ and
constraint~\eqref{eq:jinS} or~\eqref{eq:jnotinS} would be violated. Thus, $a_0$ is serving $v_0$
and $a_{N_1+1}, a_{N_1+2}, \ldots, a_{N_1+N_2}$ are serving
$v_{N_1+1},v_{N_1+2},\ldots,v_{N_1+N_2}$, respectively.  Second, it can be
verified that $v_{i}$ ($1\leq i\leq N_1$) can only be served by exactly
one RRH in $\{a_{i},\hat{a}_{i}\}$, i.e, either $\R_{v_{i}}=\{a_{i}\}$ or $\R_{v_{i}}=\{\hat{a}_{i}\}$. 
This is because, for $1\leq i\leq N_1$, the interference $w_{a_i,v_0}$ (or $w_{\hat{a}_i,v_0}$) generated from each $a_i$ (or $\hat{a}_i$) to $v_{0}$ is $3.0$, if exactly one of $\{a_{i},\hat{a}_{i}\}$ serves $v_{i}$: We assume $a_i$ serves $v_i$, then
\[
w_{a_i,v_{0}}=p_{a_i}|h_{a_i,v_0}|^2\rho_{a_i}=3.0\times \frac{2.0}{\log_2(1+|\sqrt{3}|^2)} = 3.0
\]
 In this case, one can verify that $a_0$ is fully loaded. In addition, letting any $v_i$ ($1\leq i\leq N_1$) served by both $a_{i}$ and $\hat{a}_i$ results an interference to $v_0$ being larger than $3.0$, i.e.
\[
w_{a_i,v_{0}}+w_{\hat{a}_i,v_{0}}=6.0\times\frac{2.0}{\log_2(1+|2\sqrt{3}|^2)} > 3.0 
\]
Then 
$a_0$ would be overloaded ($\rho_{a_0}>1$). Besides, for each clause $j$ ($N_1< j \leq N_2$), the three
corresponding RRHs cannot be all active in serving UEs. 
Otherwise, the RRH that is
serving the UE corresponding to this clause would be overloaded. To see this, consider a clause $(b_1\vee b_2\vee \hat{b}_3)$ and its corresponding RRHs $a_1$, $a_2$ and
$\hat{a}_3$. Assume this clause is associated with some UE $j$ ($N_1<j\leq N_1+N_2$). By the above discussion, any of $a_1$, $a_2$ and
$\hat{a}_3$ is fully loaded if it is active. Then, if all of them are active, we have 
\[
w_{a_1,v_j}=w_{a_2,v_j}=w_{\hat{a}_3,v_j}=3.0\times \frac{1}{3}\times 1.0=1.0
\]
and thus for this UE $j$ ($N_1<j\leq N_1+N_2$) we have
\begin{multline*}
\rho_{a_j} = \frac{d_{v_j}}{\log_2\left(1+\frac{p_{a_j}|h_{a_j,v_j}|^2}{w_{a_1,v_j}+w_{a_2,v_j}+w_{\hat{a}_3,v_j}+\sigma^2_{v_j}}\right)} \\ 
=\frac{1.0}{\log_2\left(1+\frac{3.0\times 1.0}{3.0\times1.0+1.0}\right)}>1.0
\end{multline*}
On the other hand, one can verify that if less than three of $a_1$, $a_2$ and
$\hat{a}_3$ are active for serving UEs, then $\rho_{a_j}\leq 1$.


Now suppose there
is an RRH-UE association that is feasible. For each Boolean variable $b_i$, we set
$b_i$ to be true if $\hat{a}_{i}$ is serving UE $v_{i}$. Otherwise, $v_{i}$ must
be served by $a_{i}$ and we set $b_i$ to false. Now we evaluate the satisfiability of each clause. For the sake of presentation, denote this clause by $(b_1\vee b_2\vee\hat{b}_3)$ as an example. The clause is
satisfied if and only if at least one of its literals being true.
As discussed above, a feasible solution for the constructed instance of \textit{MaxD} cannot
have all the three corresponding RRHs $a_1$, $a_2$, and $\hat{a}_3$ being in the status of serving UEs. 
Therefore, at least one of $a_1$, $a_2$, and $\hat{a}_3$ should be idle, meaning that the corresponding one of $b_1$, $b_2$, or $\hat{b}_3$ is set to be true.
Therefore, based a feasible solution of the constructed problem, a feasible solution of
the \textit{3-SAT} problem instance can be accordingly constructed. Conversely, suppose we have a feasible solution for a \textit{3-SAT} instance. Then we choose $a_i$ to serve $v_{i}$ if $\hat{b}_i$ is true, otherwise $\hat{a}_i$ is selected instead. Doing so satisfies all the demands for UEs $j$ ($0\leq j\leq N_1$). Furthermore, the demands of UEs $i$ ($N_1<i\leq N_1+N_2$) are satisfied as well, since at most two out of the three RRHs defined for the three literals of the clause will be serving UE. Thus the RRH-UE association is feasible for the constructed instance of \textit{MaxD}. Hence the conclusion. 
\end{proof}

 Since \textit{MaxD} is $\mathcal{NP}$-hard, one cannot expect any exact algorithm with good scalability for solving \textit{MaxD} optimally, unless $\mathcal{P}=\mathcal{NP}$. 

\section{Problem Solving}
\label{sec:solving}

In this section, we derive theoretical foundations for an algorithm solving \textit{MaxD}. We introduce some basic mathematical concepts in Section~\ref{subsec:basics}. In Section~\ref{subsec:demand_scaling} we solve \textit{MaxD} to global optimum \textit{under fixed CoMP selection}, which serves as a sub-routine for the overall algorithm. In Section~\ref{subsec:association} we give a partial optimality condition for optimizing the CoMP selection. An algorithm that alternates between CoMP selection and resource allocation is proposed in Section~\ref{subsec:algorithm}. \textcolor[rgb]{0,0,0}{We further show that considering the case $\S\subseteq\J$ enables to optimize per-user rate based on the user's priority, which enables gauging the network performance enhancement in a finer granularity in Section~\ref{subsec:priority}. The services with elastic demands could be put into the group $\S$ for QoS enhancement, while keeping the demands of other inelastic services being strictly satisfied. In Section~\ref{subsec:bound}, we give a computationally efficient bounding scheme that yields upper bound for \textit{MaxD} under two extra constraints.} 

\subsection{Basics}
\label{subsec:basics}
The following lemma shows a property of the function $\vec{f}(\bm{\mu})$. The result follows the reference~\cite[Lemma~6]{7880696}.	
\begin{lemma}
Function $\vec{f}(\bm{\mu})$ is a standard interference function (SIF) \cite{Yates:1995eh} for non-negative $\bm{\mu}$, i.e. the following properties hold:
\begin{enumerate}
	\item $\vec{f}(\bm{\mu}')\geq\vec{f}(\bm{\mu})$ if $\bm{\mu}'\geq\bm{\mu}$.
	\item $\beta\vec{f}(\bm{\mu})>\vec{f}(\beta\bm{\mu})$ $(\beta>1)$.
\end{enumerate}
\label{lma:fu}
\end{lemma}

The main property of an SIF is that
a fixed point, if exists, is unique and can be computed via fixed-point iterations.
To be more specific, a vector $\bm{\mu}$ satisfying $\bm{\mu}=\vec{f}(\bm{\mu})$, if exists, is unique and can be obtained by the iterations $\bm{\mu}^{(k)}=\vec{f}(\bm{\mu}^{(k-1)})$ ($k\geq 1$) with any $\bm{\mu}^{(0)}\in\mathbb{R}^{n}_{+}$. 


\subsection{Computing the Optimum for Given CoMP Selection}
\label{subsec:demand_scaling}
In this subsection, we show that solving \textit{MaxD} with fixed CoMP selection amounts to solving an iterated function. In mathematics, an iterated function is a function from some set to the set itself. Solving an iterated function lies on finding a fixed point of it. The problem is formulated in~\eqref{eq:max2} below, in which $\R_j$ ($j\in\J$) and $\J_i$ ($i\in\R$) are given. \textcolor[rgb]{0,0,0}{The problem in~\eqref{eq:max2} is exactly \textit{MaxD} under fixed CoMP selection $\bm{\kappa}$.}

\begin{subequations}
\begin{alignat}{2}
\quad & 
    \max\limits_{\alpha> 0,\bm{\mu}\geq\bm{0}}~\alpha  \\
    \textnormal{s.t.} \quad & \mu_j\geq\alpha f_j(\bm{\mu}) \quad j \in\S \label{eq:max2-jinS2}\\
         \quad & \mu_j\geq f_j(\bm{\mu}) \quad j \in \J\backslash\S \label{eq:max2-jnotinS2}\\
         \quad & \sum_{j\in\J_i}\mu_j\leq \bar{\rho} \quad i\in\R \label{eq:max2-barrho2}
\end{alignat}
\label{eq:max2}
\end{subequations}

\textcolor[rgb]{0,0,0}{The problem in~\eqref{eq:max2} is a subproblem of \textit{MaxD}. The motivation of deriving such a subproblem is based on the conclusion in Lemma~\ref{lma:fu} that the function $\bm{f}(\bm{\mu},\bm{\kappa})$ under fixed $\bm{\kappa}$ is an SIF in $\bm{\mu}$. Below we derive a solution method for achieving global optimum of~\eqref{eq:max2} based on the foundational properties of SIF, which justifies our decomposition approach.}
For the sake of presentation, for any $\bm{\mu}$, we denote by $H$ a function that gives the normalized maximum load, i.e.,
\begin{equation}
H:\bm{\mu} \mapsto \frac{1}{\bar{\rho}}\max_{i\in\R}\sum_{j\in\J_i}\mu_j.
\end{equation}

Before presenting the solution method for~\eqref{eq:max2}, we outline two lemmas for optimality characterization. The two lemmas enable a reformulation  that  can be solved by an iterated function. 
\begin{lemma}
$[\alpha^{*},\bm{\mu}^{*}]$ is optimal to~\eqref{eq:max2} only if $H(\bm{\mu}^{*})=1$ and $\mu^{*}_j=\alpha f_j(\bm{\mu}^{*})$ for some $j\in\S$. 	
\label{lma:necessity}
\end{lemma}
\begin{proof}
Suppose $[\alpha^{*},\bm{\mu}^{*}]$ is an optimal solution such that all inequalities strictly hold in~\eqref{eq:max2-jinS2}. Then one can increase $\alpha^{*}$ to $\alpha'=\alpha^{*}+\epsilon$. By setting $\epsilon$ to be a sufficiently small positive value, one can obtain a feasible solution $[\alpha',\bm{\mu}^{*}]$ with objective value $\alpha'$, which conflicts the assumption that $[\alpha^{*},\bm{\mu}^{*}]$ is optimal.
Therefore, there exists some $j$ ($j\in\S$) such that $\mu^{*}_j=\alpha f_j(\bm{\mu}^{*})$. 

Obviously any solution $[\alpha^{*},\bm{\mu}^{*}]$ with $H(\bm{\mu}^{*})>1$ is infeasible because at least one of~\eqref{eq:max2-barrho2} is violated. Now suppose $H(\bm{\mu}^{*})<1$. Then we increase $\bm{\mu}^{*}$ to $\bm{\mu}'=\beta\bm{\mu}^{*}$ ($\beta=1+\epsilon,~\epsilon>0$). With $\epsilon$ being sufficiently small, $\bm{\mu}'$ satisfies~\eqref{eq:max2-barrho2}. Due to the scalability of $\vec{f}(\bm{\mu})$, we have 
\begin{equation}
\vec{f}(\beta \bm{\mu}^{*})<\beta\vec{f}(\bm{\mu}^{*}).
\label{eq:fu_beta1}
\end{equation} 
Consider \eqref{eq:max2-jinS2} for $\bm{\mu}^{*}$, which reads for any $j\in\S$
\begin{equation}
\mu_j^{*}\geq \alpha f_j(\bm{\mu}^{*})\Leftrightarrow \beta\mu_j^{*}\geq\alpha\beta f_j(\bm{\mu}^{*})	
\label{eq:fu_beta2}
\end{equation}
Combining~\eqref{eq:fu_beta1} with~\eqref{eq:fu_beta2} we have
\begin{equation}
\beta\mu_j^{*}>\alpha f_j(\beta\bm{\mu}^{*})\Leftrightarrow \mu'_j>\alpha f_j(\bm{\mu}').	
\end{equation}
The same process applies for deriving $\mu_j>f_j(\bm{\mu}')$ for $j\in\J\backslash\S$.
Therefore, $\bm{\mu}'$ is a feasible solution to~\eqref{eq:max2} such that all inequalities in~\eqref{eq:max2-jinS2} and~\eqref{eq:max2-jnotinS2} strictly hold. Under $\bm{\mu}'$, one can increase $\alpha^{*}$ to $\alpha'$ as earlier in the proof, to obtain a better objective value, which conflicts with our assumption that $[\alpha^{*},\bm{\mu}^{*}]$ is optimal.

 Thus, at the optimum of \textit{MaxD} there is at least one UE $j$ ($j\in\S$) such that $\mu^{*}_j=\alpha^{*}f_j(\bm{\mu}^{*})$ with $H(\bm{\mu}^{*})=1$.
\end{proof}
By Lemma~\ref{lma:necessity} we know that a solution is optimal to \eqref{eq:max2} only if there exists a fully loaded RRH. Intuitively, if all RRHs in the C-RAN have unused time-frequency resource, then one can improve the objective function such that more bits would be delivered to UEs.

\begin{lemma}
$[\alpha^{*},\bm{\mu}^{*}]$ is optimal to \eqref{eq:max2} if \eqref{eq:max2-jinS2} and~\eqref{eq:max2-jnotinS2} all hold as equality and $H(\bm{\mu}^{*})=1$.	 
\label{lma:sufficiency}
\end{lemma}
\begin{proof}
Suppose \eqref{eq:max2-jinS2} and~\eqref{eq:max2-jnotinS2} hold for all $j\in\J$ as equalities with $[\alpha^{*},\bm{\mu}^{*}]$. Consider any $\alpha'$ ($\alpha'>\alpha^{*}$). Replacing $\alpha^{*}$ by $\alpha'$ in~\eqref{eq:max2-jinS2} causes  \eqref{eq:max2-jinS2} being violated. Thus $\mu^{*}_j$ ($j\in\S$) must increase to have \eqref{eq:jinS2} remains satisfied. Then 
$\vec{f}(\bm{\mu})$ would grow due to its monotonicity, resulting in the violations of~\eqref{eq:max2-jnotinS2}. Therefore, to have \eqref{eq:max2-jinS2} and~\eqref{eq:max2-jnotinS2} remain satisfied, the vector $\bm{\mu}$ must be increased. Denote by $\bm{\mu}'$ the newly obtained resource allocation.
Since $H(\bm{\mu}^{*})=1$, then we must have $H(\bm{\mu}')>1$ which violates some constraint in \eqref{eq:max2-barrho2}. Hence the conclusion.
\end{proof}

Next we derive a solution method that achieves the global optimum of~\eqref{eq:max2}. 
We define function $\vec{F}_{\alpha}$ as follows.
\begin{equation}
\vec{F}_{\alpha}: \bm{\mu}\mapsto \left[\frac{f_1(\bm{\mu})}{\pi_1(\alpha)},\frac{f_2(\bm{\mu})}{\pi_2(\alpha)},\ldots,\frac{f_n(\bm{\mu})}{\pi_n(\alpha)} \right]	
\end{equation}
where
\begin{equation}
\pi_j(\alpha) = \left\{
\begin{array}{ll}
1	&   j\in\S \\
\alpha & \textnormal{otherwise}
\end{array}
\right.
\end{equation}
Note that for any given $\alpha>0$, the function $\vec{F}_{\alpha}$ is an SIF in $\bm{\mu}$. The problem in~\eqref{eq:max2} can be reformulated below.
\begin{equation}
\max_{\alpha\geq 1,\bm{\mu}\geq\bm{0}}\alpha~\textnormal{s.t.}~\alpha\vec{F}_{\alpha}(\bm{\mu})\leq\bm{\mu},H(\bm{\mu})=1.
\label{eq:max3}
\end{equation}

\textcolor[rgb]{0,0,0}{The recursive equations in Theorem~\ref{thm:compute_G} below give the solution method for solving~\eqref{eq:max3} (and equivalently~\eqref{eq:max2}). The optimality of this method is guaranteed by Theorem~\ref{thm:opt}.} The proofs of both Theorem~\ref{thm:compute_G} and Theorem~\ref{thm:opt} are in the Appendix. The symbol ``$\circ$'' denotes the function composition, i.e. $g_1\circ g_2(\mathtt{var}) = g_1(g_2(\mathtt{var}))$. 
\begin{theorem}
Denote $[\alpha^{*},\bm{\mu}^{*}]=\lim_{k\rightarrow\infty}[\alpha^{(k)},\bm{\mu}^{(k)}]$, where 
\begin{equation}
\alpha^{(k)}=\frac{1}{H\circ\vec{F}_{\alpha^{(k-1)}}(\bm{\mu}^{(k-1)})},~k\geq 1
\label{eq:G_iter}
\end{equation}
and
\begin{equation}
\bm{\mu}^{(k)}=\frac{\vec{F}_{\alpha^{(k-1)}}(\bm{\mu}^{(k-1)})}{H\circ\vec{F}_{\alpha^{(k-1)}}(\bm{\mu}^{(k-1)})},~k\geq 1,
\label{eq:G_iter2}
\end{equation}
with $\alpha^{(0)}>0$ and $\bm{\mu}^{(0)}\in\mathbb{R}^{n}_{+}$. Then $H(\bm{\mu}^{*})=1$ holds.
\label{thm:compute_G}
\end{theorem}
\begin{theorem}
$[\alpha^{*},\bm{\mu}^{*}]$ in Theorem~\ref{thm:compute_G} is optimal to~\eqref{eq:max2}.	
\label{thm:opt}
\end{theorem}

Theorem~\ref{thm:compute_G} and Theorem~\ref{thm:opt} guarantee that for an arbitrary set of UEs $\S$ in the network, one can iteratively compute the maximum demand scaling factor $\alpha^{*}$ for $\S$. As a special case, when $\pi_j(\alpha)=1$ for all $j\in\J$, solving the problem in~\eqref{eq:max3} is to find the eigenvalue $\bm{\mu}$ and the eigenvector $1/\alpha$ of the equation system $\vec{f}(\bm{\mu})=(1/\alpha)\bm{\mu}$ such that $H(\bm{\mu})=1$. \textcolor[rgb]{0,0,0}{The iterative solution method in Theorem~\ref{thm:compute_G} serves as a subroutine for solving \textit{MaxD}, shown in Section~\ref{subsec:algorithm} later.}

\subsection{CoMP Selection Optimization}
\label{subsec:association}

We show a lemma for optimizing the CoMP selection.
The detailed proof of the lemma is based on~\cite[Theorem 3]{7343476}. The following notations are introduced. We consider two matrices $\bm{\kappa}$ and $\bm{\kappa}'$ with the following relationship: Each column in $\bm{\kappa}$ has at least one non-zero element (meaning that every UE is associated to at least one RRH); There exists exactly one RRH-UE pair ($i,j$) such that $\kappa_{ij}=0$ and $\kappa'_{ij}=1$, respectively; For any other pair ($r,q$), $\kappa'_{rq}=\kappa_{rq}$. Note that the only difference between the two associations $\bm{\kappa}$ and $\bm{\kappa}'$ is that $\bm{\kappa}'$ includes the CoMP link from RRH $i$ to UE $j$ while $\bm{\kappa}$ does not. 
Denote by $\bm{\mu}^{*}$ and $\bm{\mu}'^{*}$ the corresponding optimal resource allocations obtained respectively by $\bm{\kappa} $ and $\bm{\kappa}'$ (i.e. $\bm{\mu}^{*}=\vec{F}_{\alpha}(\bm{\mu}^{*},\bm{\kappa})$ and $\bm{\mu}'^{*}=\vec{F}_{\alpha}(\bm{\mu}'^{*},\bm{\kappa}')$). 
For the sake of presentation, given some $\alpha$ ($\alpha>0$), we represent $\rho_i$ as a function of the resource allocation $\bm{\mu}$ and the RRH-UE association $\bm{\kappa}$:
\begin{equation}
\rho_i:[\bm{\mu},\bm{\kappa}]\mapsto \sum_{j\in\J_i(\bm{\kappa})} F_{\alpha,j}(\bm{\mu},\bm{\kappa}).
\end{equation}
Consider two sequences $\bm{\mu}^{(0)},\bm{\mu}^{(1)},\ldots,\bm{\mu}^{(\infty)}$ and $\bm{\rho}^{(0)},\bm{\rho}^{(1)},\ldots,\bm{\rho}^{(\infty)}$, where $\bm{\mu}^{(k)}=\vec{f}(\bm{\rho}^{(k-1)},\bm{\kappa}')$, $\bm{\rho}^{(k)}=\bm{\rho}(\bm{\mu}^{(k-1)},\bm{\kappa})$ for $k\geq 1$, with $\bm{\mu}^{(0)}=\bm{\mu}^{*}$ and $\bm{\rho}^{(0)}=\bm{\rho}(\bm{\mu}^{(0)},\bm{\kappa})$. The convergence of the two sequences is guaranteed as the function $\vec{F}_{\alpha}$ with fixed $\alpha$ is an SIF in both $\bm{\mu}$ and $\bm{\rho}$ \cite{7343476}. We provide the following lemma.
\begin{lemma}
$\bm{\rho}(\bm{\mu}'^{*},\bm{\kappa}')\leq\bm{\rho}(\bm{\mu}^{*},\bm{\kappa})$ if for any $k\geq 1$ we have $\rho_{i}(\bm{\mu}^{(k)},\bm{\kappa}')\leq \rho_{i}^{(k)}$. 
\label{lma:adding}
\end{lemma}

Lemma~\ref{lma:adding} serves as a sufficient condition for RRH load improvement by CoMP. Specifically, in order to check whether adding a CoMP link between RRH $i$ and UE $j$ would reduce the load levels of RRHs, we iteratively construct the two sequences $\bm{\mu}^{(0)},\bm{\mu}^{(1)},\ldots,\bm{\mu}^{(\infty)}$ and $\bm{\rho}^{(0)},\bm{\rho}^{(1)},\ldots,\bm{\rho}^{(\infty)}$. Once there exists $k\geq 1$ such that $\rho_{i}(\bm{\mu}^{(k)},\bm{\kappa}')\leq \rho_{i}^{(k)}$, we conclude that the bits demand by all UEs can be satisfied with lower time-frequency resource consumption (i.e. the load of RRH) under the association $\bm{\kappa}'$ than $\bm{\kappa}$. In Section~\ref{subsec:algorithm} below, we show that the condition in Lemma~\ref{lma:adding} can be incorporated with the solution method in Theorem~\ref{thm:compute_G}, to form our joint optimization algorithm.

\subsection{Algorithm}
\label{subsec:algorithm}

The theoretical properties derived in Section~\ref{subsec:demand_scaling} and Section~\ref{subsec:association} enable an algorithm for \textit{MaxD}. To be more specific, given any RRH-UE association $\bm{\kappa}$, one can obtain the corresponding maximum scaling factor $\alpha^{*}$ together with $\bm{\mu}^{*}$ iteratively by~\eqref{eq:G_iter} and~\eqref{eq:G_iter2} in Theorem~\ref{thm:compute_G}. At the convergence, $H(\bm{\mu}^{*})=1$ holds by Theorem~\ref{thm:compute_G}. Taking one step further, if we fix $\alpha^{*}$ and consider all the candidate RRHs for each UE, 
\textcolor[rgb]{0,0,0}{using
Lemma~\ref{lma:adding} (in Section~\ref{subsec:association}) enables us to determine whether a new association $\bm{\kappa}'$ that includes some newly added CoMP link would lead to load improvement. If yes, the corresponding resource allocation $\bm{\mu'}$ under $\bm{\kappa}'$ must have $H(\bm{\mu}')<1$.  Then by Lemma~\ref{lma:necessity}, the current solution $[\alpha^{*},\bm{\mu}']$ is not optimal. 
Once the load is improved, applying \eqref{eq:G_iter} and~\eqref{eq:G_iter2} under $\bm{\kappa}'$ again guarantees an overall improvement.
We remark that Algorithm~\ref{alg:MaxD} works in an online manner in terms of the candidate RRH-UE pairs. To be specific, in each iteration of the outer loop, once a new CoMP link is added, the demand scaling factor $\alpha$ is guaranteed to be increased by the end of this iteration, based on the discussion above.
}
This process is detailed in Algorithm~\ref{alg:MaxD}.
\begin{algorithm}[!ht]
\KwIn{$\bm{\kappa}^{(0)},\S,\epsilon>0$, $\bar{\rho}$}
\KwOut{$\bm{\kappa}^{*},\bm{\mu}^{*},\alpha^{*}$}
$\alpha^{(0)}\leftarrow 1$; (Or other positive number) \\
$\bm{\mu}^{(0)}\leftarrow \vec{0}$; (Or other non-negative vector) \\
\Repeat{$\bm{\kappa}^{(c)}=\bm{\kappa}^{(c-1)}$}{\label{l1:start}
	\For{$i\leftarrow 1$ \textnormal{to} $m$, $j\leftarrow 1$ \textnormal{to} $n$, \textnormal{with} $\kappa_{ij}^{(c)}=0$}{
		$k\leftarrow 1$\;
		\Repeat{$\norm{\bm{\mu}^{(k)}-\bm{\mu}^{(k-1)}}<\epsilon$}{\label{l3:start}
			$\bm{\kappa}^{(c)}\leftarrow\bm{\kappa}^{(c-1)}$\;
			$\bm{\kappa}'\leftarrow\bm{\kappa}^{(c)}$\;
			$\kappa'_{ij}\leftarrow 1$\;\label{step:add}
			$\bm{\rho}^{(k)}\leftarrow\bm{\rho}(\bm{\mu}^{(k-1)},\bm{\kappa}^{(c)})$\;\label{step:eval}
			$\bm{\mu}^{(k)}\leftarrow\vec{F}_{\alpha}(\bm{\rho}^{(k)},\bm{\kappa}')$\;
			\uIf{$\rho_i(\bm{\mu}^{(k)},\bm{\kappa}')\leq\rho_i^{(k)}$}{\label{l3:adding}
				$\bm{\kappa}^{(c)}\leftarrow\bm{\kappa}'$\;
			}
			$k\leftarrow k+1$\;
		}\label{l3:end}
		$h\leftarrow 1$\;\label{l:alpha0}
		\Repeat{$\norm{[\alpha^{(h)},\bm{\mu}^{(h)}]-[\alpha^{(h-1)},\bm{\mu}^{(h-1)}]}<\epsilon$}{\label{l3:start2}
			$\alpha^{(h)}\leftarrow\frac{1}{H\circ\vec{F}_{\alpha^{(h-1)}}(\bm{\mu}^{(h-1)},\bm{\kappa}^{(c)})}$\;
			$\bm{\mu}^{(h)}\leftarrow\frac{\vec{F}_{\alpha^{(h-1)}}(\bm{\mu}^{(h-1)},\bm{\kappa}^{(c)})}{H\circ\vec{F}_{\alpha^{(h-1)}}(\bm{\mu}^{(h-1)},\bm{\kappa}^{(c)})}$\;
			$h\leftarrow h+1$\;
		}\label{l3:end2}
		$\alpha\leftarrow\alpha^{(h)}$\;
		$\bm{\mu}\leftarrow\bm{\mu}^{(h)}$\;
	}
	$c\leftarrow c+1$\;
}\label{l1:end}
$\bm{\kappa}^{*}\leftarrow\bm{\kappa}^{(c)}$; $\bm{\mu}^{*}\leftarrow\bm{\mu}$; $\alpha^{*}\leftarrow\alpha$\;
\caption{Joint Demand Scaling and CoMP Selection}
\label{alg:MaxD}
\end{algorithm}

The input of Algorithm~\ref{alg:MaxD} consists of an initial RRH-UE association $\bm{\kappa}^{(0)}$, a set $\S$ ($\S\subseteq\J$) of UEs for demand scaling, and a positive value $\epsilon$ that is the tolerance of convergence. For the output, Algorithm~\ref{alg:MaxD} gives the optimized RRH-UE association $\bm{\kappa}^{*}$,  the corresponding optimal resource allocation $\bm{\mu}^{*}$, and the demand scaling factor $\alpha^{*}$. The algorithm goes through all the candidate RRH-UE pairs for CoMP. For each candidate RRH-UE pair, the algorithm applies the partial optimality condition in Lemma~\ref{lma:adding}. Specifically, if the condition in Line~\ref{l3:adding} is satisfied for any RRH-UE pair ($i,j$), then adding a CoMP link ($i,j$) improves the load levels.  When the loop in Lines \ref{l3:start}--\ref{l3:end} ends, the newly optimized association $\bm{\kappa}^{(c)}$ is obtained. The computational method in Theorem~\ref{thm:compute_G} is implemented in Lines~\ref{l:alpha0}--\ref{l3:end2}, by which we obtain the optimal demand scaling for the set $\S$ and the corresponding resource allocation to satisfy the scaled demands under $\bm{\kappa}^{(c)}$.

The complexity of Algorithm~\ref{alg:MaxD} is analyzed as follows. For simplicity, we assume the operations on vectors are atomic and can be done in $O(1)$. By~\cite{6353394} (along with our proof for Theorem~\ref{thm:compute_G}), the iterations in Lines~\ref{l3:start}--\ref{l3:end} and Lines~\ref{l3:start2}--\ref{l3:end2} have linear convergence such that the complexity of the two loops (i.e. the required number of iterations) with tolerance $\epsilon$ is $O(\log\frac{1}{\epsilon})$~\cite[Page~37]{nesterov2013introductory}. Besides, the outer for-loop runs in $O(mn)$, and is executed with a maximum of $m\times n$ times\footnote{Because there are at most $m\times n$ links that can be added.}. Hence the total complexity is $O(m^2n^2\log\frac{1}{\epsilon})$. \textcolor[rgb]{0,0,0}{We remark that our algorithm is scalable from the computational theory perspective, since it is polynomial in the number of RRHs and UEs.} The tolerance parameter $\epsilon$ can be selected in an on-demand manner: the smaller the value of $\epsilon$, the higher the accuracy of the algorithm, meanwhile requiring more iterations. Numerically, the impact of the tolerance parameter $\epsilon$ on the algorithm convergence is discussed later in Section~\ref{subsec:convergence}.

\subsection{\textcolor[rgb]{0,0,0}{Priority-aware Per-user Rate Optimization}}
\label{subsec:priority}
\textcolor[rgb]{0,0,0}{We show in this subsection how Algorithm~\ref{alg:MaxD} applies to optimizing per-user rate by taking into account priority. For the sake of presentation, we impose $\K$ to be an ordered set consisting of candidate UEs for rate enhancement. The UEs in $\K$ are arranged in descending order according to their priority. Without loss of generality, suppose $\K = \{1,2,\ldots K\}$. For each UE $k$ in $\K$, we assign a budget $\Delta_k$ for resource allocation. That is, $\Delta_k$ is the affordable increase of the maximum RRH load level for enhancing the rate of UE $k$. Denote by $\bar{\rho}_0$ the maximum RRH load level initially.}

\textcolor[rgb]{0,0,0}{The optimization procedure is as follows. We first run Algorithm~\ref{alg:MaxD} with $\bar{\rho}=\bar{\rho}_0+\Delta_1$ and $\S=\{1\}$. At convergence, the maximum load of RRHs equals $\Delta_1$, and the rate of UE $1$ after optimization is $\alpha^{*}d_1$. We then replace the original base demand $d_1$ by $\alpha^{*}d_1$ for UE 1, and run Algorithm~\ref{alg:MaxD} with $\bar{\rho}=\bar{\rho}_0+\Delta_1+\Delta_2$ and $\S=\{2\}$ for UE 2, and so on. The process repeats until we reach UE $K$. At the end, the delivered demand of all UEs in $\K$ would be enhanced, \textcolor[rgb]{0,0,0}{with the maximum RRH load being $\bar{\rho}=\rho_0 + \sum_{k\in\K}\Delta_k$}.}

\subsection{\textcolor[rgb]{0,0,0}{Upper Bound for \textit{MaxD} under Two Extra Constraints}}
\label{subsec:bound}
\textcolor[rgb]{0,0,0}{
We derive an upper bound for \textit{MaxD} under two extra constraints on CoMP selection. Denote by $\J^{+}_i$ the RRH $i$'s candidate set of served UEs, \textcolor[rgb]{0,0,0}{i.e., the set of UEs that may potentially be served by RRH $i$}. The first constraint is as follows.}
\textcolor[rgb]{0,0,0}{
\begin{equation}
\textnormal{Candidate UE constraints: $\kappa_{ij}=0$ if $j\notin\J^{+}_i$, $i\in\I$.}	
\label{eq:constr_UE}
\end{equation}
}
\textcolor[rgb]{0,0,0}{We impose that for each UE, there is a dominant RRH (e.g. the one with the strongest signal or with the shortest distance to the UE etc.) that serves the UE all the time. We call this dominant RRH the home RRH.} Denote by $\J^{-}_i$ the set of UEs of which the home RRH is $i$. The second constraint is below.
\textcolor[rgb]{0,0,0}{
\begin{equation}
\textnormal{Candidate RRH constraints: $\kappa_{ij}=1$ if $j\in\J^{-}_i$, $i\in\I$.}
\label{eq:constr_RRH}
\end{equation}
}

\textcolor[rgb]{0,0,0}{Define a matrix $\check{\bm{\kappa}}$, where for any RRH $i$ ($i\in\R$) entry $\check{\kappa}_{ij}=1$ if $j\in\J^{+}_i$. 
Consider the optimization problem below.
\begin{subequations}
\begin{alignat}{2}
\quad & 
    \max\limits_{\alpha> 0,\bm{\mu}\geq\bm{0}}~\alpha  \\
    \textnormal{s.t.} \quad & \mu_j\geq\alpha f_j(\bm{\mu},\check{\bm{\kappa}}) \quad j \in\S \label{eq:jinS2}\\
         \quad & \mu_j\geq f_j(\bm{\mu},\check{\bm{\kappa}}) \quad j \in \J\backslash\S \label{eq:jnotinS2}\\
         \quad & \sum_{j\in\J^{-}_i}\mu_j\leq \bar{\rho} \quad i\in\R \label{eq:barrho2}
\end{alignat}
\label{eq:upper_bound}
\end{subequations}
}
\textcolor[rgb]{0,0,0}{
\begin{theorem}
Solving~\eqref{eq:upper_bound} yields an upper bound of $\alpha$ for \textit{MaxD} with constraints \eqref{eq:constr_UE} and \eqref{eq:constr_RRH}.
\end{theorem}
\begin{proof}
The proof is based on \cite[Lemma 12]{7880696}. The derivation below is under under constraints \eqref{eq:constr_UE} and \eqref{eq:constr_RRH}. By \cite[Lemma 12]{7880696}, we conclude that for any CoMP selection $\bm{\kappa}$ , we always have $f_j(\bm{\mu},\bm{\kappa})\geq f_j(\bm{\mu},\check{\bm{\kappa}})$ for any $\bm{\mu}\geq\bm{0}$. In addition, $\sum_{j\in\J_i^{-}}\mu_j\leq \sum_{j\in\J_i}\mu_j$ holds for any $\bm{\mu}\geq\bm{0}$ and any $\bm{\kappa}$ (note that $\J_i$ is related to $\bm{\kappa}$). Therefore, \eqref{eq:upper_bound} is indeed a relaxation of \textit{MaxD} (with \eqref{eq:constr_UE} and \eqref{eq:constr_RRH}), such that solving the former always yields a better (or at least no worse) objective function value than the latter at optimum. Hence the conclusion. 
\end{proof}
}

\textcolor[rgb]{0,0,0}{
We remark that the upper bound is derived for \textit{MaxD} under the two extra constraints \eqref{eq:constr_UE} and \eqref{eq:constr_RRH}, though not proved theoretically to be an exact upper bound of the original problem \textit{MaxD}. On the other hand, with the two constraints, solving \eqref{eq:upper_bound} for obtaining the bound is quite straightforward and computationally efficient.
}

We also remark that, for any specific UE $j$, letting $\S=\{j\}$ and solving the corresponding formulation~\eqref{eq:upper_bound} yield the upper bound of the satisfiable demand of UE $j$. In addition, denote by $\bm{d}_{\S}$ the demands of $\S$ before scaling. Then $|\S|$ is an upper bound for the number of users with demands being no less than $\alpha^{*}\bm{d}_{\S}$, under the worst channel conditions of $\S$.

\section{Simulation}
\label{sec:numerical}

The C-RAN under consideration consists of one hexagonal region, within which multiple UEs and RRHs are randomly deployed. The RRHs are coordinated by the cloud and cooperate with each other for CoMP transmission.
Initially, no UE is in CoMP, and each UE is served by the RRH with the best signal power. The network layout is illustrated in~\figurename~\ref{fig:10scaling}. Parameter settings are given in \tablename~\ref{tab:sim}.
\textcolor[rgb]{0,0,0}{The user demands setting is configured for doing performance benchmarking for the cases with high RRH loads.
In our simulation, the user demand subject to scaling is initially uniform. We use the non-CoMP case as the baseline. In the non-CoMP case, each UE is served with its single RRH that is the home RRH. The user demand is set such that for the baseline, with $\alpha=1.0$, there is at least one RRH $i$ ($i\in\R$) reaching the load limit ($\rho_i=\bar{\rho}$). This demand\footnote{
\textcolor[rgb]{0,0,0}{
The method of obtaining this user demand is as follows. For any initialized $d_j$ ($j\in\J$), applying Theorem~5 directly yields the corresponding demand scaling factor $\alpha$, which leads to that at least one RRH $i$ ($i\in\R$) reaching the load limit. Our expected user demand equals the initial demand multiplied by this $\alpha$.}
} is normalized by $M\times B$. 
}
For clarity, we let $\J_{\text{(CoMP)}}$ denote the set of CoMP UEs in the solution obtained from Algorithm~\ref{alg:MaxD}. The performance is studied by using four metrics for evaluation, defined below, referred to as Metric~\ref{item:metric-alpha}), Metric~\ref{item:metric-CoMP}), Metric~\ref{item:metric-delivery}), and Metric~\ref{item:metric-CoMP-SR}), respectively.  
\begin{table}[!ht]
\centering
\caption{Simulation Parameters.}
\begin{tabular}{ll}
\toprule
\textbf{Parameter} & \textbf{Value} \\
Hexagon radius & \textcolor[rgb]{0,0,0}{$\{500, 100\}$ m} \\
Carrier frequency & $2$ GHz \\
Total bandwidth & $20$ MHz\\
Number of UEs ($|\J|$) & \textcolor[rgb]{0,0,0}{$\{100,570\}$} \\
Number of RRHs ($|\R|$) & \textcolor[rgb]{0,0,0}{$\{10,57\}$} \\
Path loss & COST-231-HATA \\
Shadowing (Log-normal) & $3$ dB standard deviation\\
Fading & Rayleigh flat fading \\
Noise power spectral density & $-173$ dBm/Hz \\
RB bandwidth & $180$ KHz \\
Transmit power on one RB & $400$ mW  \\
User demand distribution & Uniform  \\
Convergence tolerance ($\epsilon$) & $10^{-4}$ \\
\textcolor[rgb]{0,0,0}{Demand scaling proportion ($|\S|/|\J|$)} & \textcolor[rgb]{0,0,0}{$\{10,20,40,60,80,100\}$ (\%) }\\
\bottomrule
\end{tabular}
\label{tab:sim}	
\end{table}
\begin{enumerate}
\item 
\label{item:metric-alpha}
\textit{Improvement of $\alpha$}: The metric is the objective function of \textit{MaxD}. It reflects the capacity improvement of $\S$ (i.e. the user-centric performance).
\item 
\label{item:metric-CoMP}
$|\J_{(\text{(CoMP)})}|$: This metric is the number of UEs involved in CoMP. It is used to relate the amount of CoMP to the capacity improvement. 
\item 
\label{item:metric-delivery}
\textit{Increase (abbreviated as ``inc'') of delivered demand}: The metric
is the amount of relative increase of the total delivered demand.
It reflects the capacity improvement of $\J$, (i.e., the network-wise performance).
\item 
\label{item:metric-CoMP-SR}
$|\J_{\text{(CoMP)}}\cap\S|$: This metric is the number of CoMP UEs inside $\S$. The metric is used for examining how CoMP for UEs inside $\S$ affects the demand scaling performance. One can also infer the number of CoMP UEs outside $\S$ by this metric together with Metric~\ref{item:metric-CoMP}).
\end{enumerate}

In the remaining parts of the section, we show the four metrics as functions of the number of UEs in $\S$ (i.e., $|\S|$), and the RRH load limit $\bar{\rho}$. Since our proposed algorithm employs a user-centric strategy for demand scaling, as for comparison, we refer to the strategy of scaling up demand for all UEs as a fairness-based strategy.

\subsection{Performance with Respect to $|\S|$ (with $\bar{\rho}=1.0$)}
\label{subsec:sim-S/J}

\begin{figure*}[t] 
\centering 
\subfigure[$|\S|=20$\label{fig:10scaling-20}]{\includegraphics[width=0.24\linewidth]{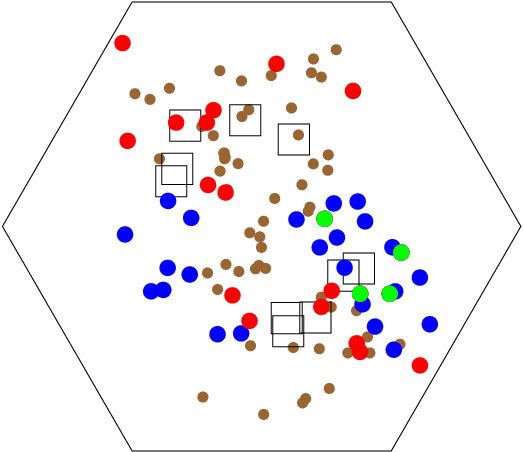}}
\subfigure[$|\S|=40$\label{fig:10scaling-40}]{\includegraphics[width=0.24\linewidth]{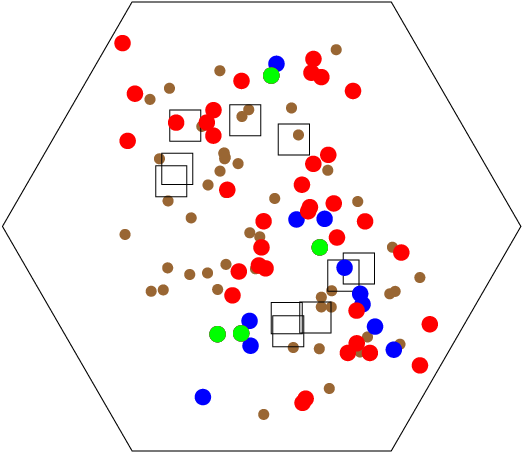}}
\subfigure[$|\S|=60$\label{fig:10scaling-60}]{\includegraphics[width=0.24\linewidth]{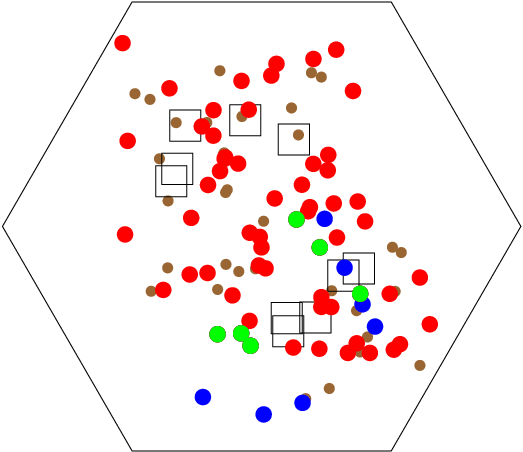}}
\subfigure[$|\S|=80$\label{fig:10scaling-80}]{\includegraphics[width=0.24\linewidth]{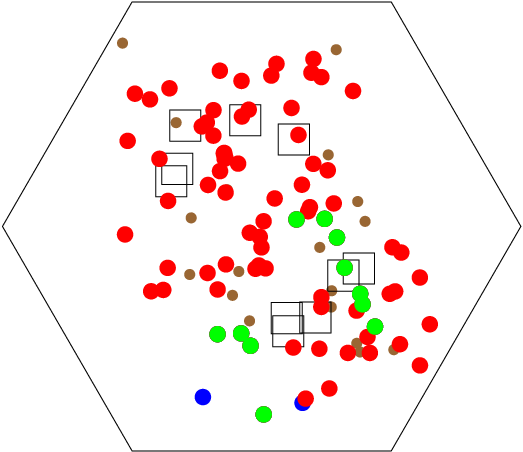}}
\caption{This figure shows some snapshots of the optimized CoMP selections. The UEs and the RRHs are illustrated by dots and rectangles, respectively. The number of UEs in $\S$ are $20$, $40$, $60$, and $80$, respectively. The non-CoMP UEs in $\S$ are marked red. The CoMP UEs in $\S$ are marked green. The CoMP UEs in $\J\backslash\S$ are marked blue. \textcolor[rgb]{0,0,0}{The other UEs are in brown.} Note that the locations of RRHs are randomly generated within the hexagonal region in all our simulations. The figure is obtained from one simulation and is representative for all the simulation results. } 
\label{fig:10scaling} 
\end{figure*}

\begin{figure}[t]
\begin{tikzpicture}
\begin{groupplot}[
    group style={
        group name=my plots,
        group size=1 by 4,
        xlabels at=edge bottom,
        xticklabels at=edge bottom,
        vertical sep=6pt
    },
	xlabel={Number of scaled UEs $|\S|$},
	label style = {font=\fontsize{9pt}{10pt}\selectfont},
	legend cell align={left},
	legend pos = north west,
	legend style = {font=\fontsize{8pt}{10pt}\selectfont},
	axis background/.style={fill=white},
	minor x tick num=4,
	minor y tick num=4,
	major tick length=0.15cm,
	minor tick length=0.075cm,
	tick style={semithick,color=black},
	height=0.42\linewidth,
	width=0.9\linewidth,
]
\nextgroupplot[ylabel style={align=center}, ylabel={Metric \ref{item:metric-alpha}) \\ Improvement \\ of $\alpha$ (\%) }]
\addplot  [smooth, mark= diamond, color=red] 
		   coordinates{(10, 27.824) (20, 21.2774) (40, 16.483) (60, 14.3662) (80, 13.2022) (100, 13.0022)};
\nextgroupplot[ylabel style={align=center}, ylabel={Metric \ref{item:metric-CoMP}) \\ $|\J_{\text{(CoMP)}}|$ \\ ~}]
\addplot  [smooth, mark=square, color=brown] 
		   coordinates{(10, 16.28) (20, 13.52) (40, 12.16) (60, 11.36) (80, 11.12) (100, 11.11)};
\nextgroupplot[ylabel style={align=center}, ylabel={Metric \ref{item:metric-delivery}) \\ Inc of delivered\\ demand (\%)}]
\addplot  [smooth, mark=*, color=purple]
		   coordinates{(10, 2.7824) (20, 4.25548) (40, 6.5932) (60, 8.61972) (80, 10.5618) (100, 13.0022)};
\nextgroupplot[ylabel style={align=center}, ylabel={Metric \ref{item:metric-CoMP-SR}) \\ $|\J_{\text{(CoMP)}}\cap\S|$ \\ ~}]
\addplot  [smooth, mark=triangle*, color=blue] 
		   coordinates{(10, 1.29264) (20, 2.63535) (40, 4.3798) (60, 7.13426) (80, 8.50269) (100, 11.11)};
\end{groupplot}
\end{tikzpicture}
\caption{This figure shows the four metrics in function of $|\S|$. The non-CoMP case is used as baseline for Metrics~\ref{item:metric-alpha}) and~\ref{item:metric-delivery}).}
\label{fig:comp}
\end{figure}

\pgfplotsset{compat=1.11,
    /pgfplots/ybar legend/.style={
    /pgfplots/legend image code/.code={%
       \draw[##1,/tikz/.cd,yshift=-0.25em]
        (0cm,0cm) rectangle (3pt,0.8em);},
   },
}
\begin{figure}[t]
\begin{tikzpicture}
\begin{groupplot}[
    group style={
        group name=my plots,
        group size=1 by 4,
        xlabels at=edge bottom,
        xticklabels at=edge bottom,
        vertical sep=6pt
    },
	label style = {font=\fontsize{9pt}{10pt}\selectfont},
	legend style = {font=\fontsize{6pt}{10pt}\selectfont},
    ybar,
    enlarge x limits=0.2,
    legend style={
        anchor=north,legend columns=-1
     },
    legend pos = north west,
    ylabel={Improvement of $\alpha$ (\%)},
    xlabel={Load limit $\bar{\rho}$},
    symbolic x coords={0.4, 0.6, 0.8, 1.0},
    xtick=data,
	minor x tick num=4,
	minor y tick num=4,
	major tick length=0.15cm,
	minor tick length=0.075cm,
	tick style={semithick,color=black},
	height=0.42\linewidth,
	width=0.9\linewidth,
    xtick align=inside,
    ymin = 0,
    ]
\nextgroupplot[ylabel style={align=center}, ylabel={Metric \ref{item:metric-alpha}) \\ Improvement \\ of $\alpha$ (\%) }, bar width = 0.15cm, 
    legend entries={$|\S|=20$,$|\S|=40$,$|\S|=60$,$|\S|=80$},
    legend columns = 4,
]
\addplot  coordinates {(0.4, 1.57247) (0.6, 5.34643)  (0.8, 13.4273) (1.0, 21.2774) };
\addplot  coordinates {(0.4, 0.889891) (0.6, 3.64669)  (0.8, 10.0093) (1.0, 16.483) };
\addplot  coordinates {(0.4, 0.332242) (0.6, 3.31965)  (0.8, 9.59463) (1.0, 14.3662) };
\addplot  coordinates {(0.4, 0.181444) (0.6, 2.8877)  (0.8, 8.16005) (1.0, 13.2022) };

\nextgroupplot[ylabel style={align=center}, ylabel={Metric \ref{item:metric-CoMP}) \\ $|\J_{\text{(CoMP)}}|$ \\ ~},bar width = 0.15cm]
\addplot  coordinates {(0.4, 2.46) (0.6, 6.18)  (0.8, 11.46) (1.0, 13.26) };
\addplot  coordinates {(0.4, 1.26) (0.6, 4.56)  (0.8, 9.42) (1.0, 12.06) };
\addplot  coordinates {(0.4, 0.84) (0.6, 4.48)  (0.8, 9.12) (1.0, 11.6) };
\addplot  coordinates {(0.4, 0.32) (0.6, 4.4)  (0.8, 8.48) (1.0, 11.6) };

\nextgroupplot[ylabel style={align=center}, ylabel={Metric \ref{item:metric-delivery}) \\ Inc of delivered \\ demand (\%)},bar width = 0.15cm]
\addplot  coordinates {(0.4, 0.314494) (0.6, 1.06929)  (0.8, 2.68546) (1.0, 4.25548) };
\addplot  coordinates {(0.4, 0.355956) (0.6, 1.45867)  (0.8, 4.00373) (1.0, 6.5932) };
\addplot  coordinates {(0.4, 0.199345) (0.6, 1.99179)  (0.8, 5.75678) (1.0, 8.61972) };
\addplot  coordinates {(0.4, 0.145155) (0.6, 2.31016)  (0.8, 6.52804) (1.0, 10.5618) };

\nextgroupplot[ylabel style={align=center}, ylabel={Metric \ref{item:metric-CoMP-SR}) \\ $|\J_{\text{(CoMP)}}\cap\S|$ \\ ~ }  ,bar width = 0.15cm]
\addplot  coordinates {(0.4, 0.46) (0.6, 1.24)  (0.8, 2.28) (1.0, 2.58) };
\addplot  coordinates {(0.4, 0.54) (0.6, 1.66)  (0.8, 3.68) (1.0, 5.4) };
\addplot  coordinates {(0.4, 0.32) (0.6, 2.56)  (0.8, 5.1) (1.0, 6.86) };
\addplot  coordinates {(0.4, 0.26) (0.6, 3.58)  (0.8, 6.62) (1.0, 9.26) };

\end{groupplot}
\end{tikzpicture}
\caption{This figure shows the four metrics in function of $\bar{\rho}$. The non-CoMP case is baseline for Metrics~\ref{item:metric-alpha}) and~\ref{item:metric-delivery}). The legend applies to all subfigures.}
\label{fig:rhobar}
\end{figure}

By observing \figurename~\ref{fig:comp}, as expected for Metric~\ref{item:metric-alpha}), one can achieve more improvement of $\alpha$ by CoMP when the size of $\S$ is smaller. The reason is very understandable: With the same amount of resource, enhancing the performance for a small group of UEs is generally easier than for a larger group. CoMP achieves considerable improvement of $\alpha$, ranging from $13\%$ to $28\%$. 

The user-centric performance benefits from CoMP through both direct and indirect effects. For explanation, we use
\figurename~\ref{fig:10scaling} as an illustration. \figurename~\ref{fig:10scaling} shows some snapshots of our experiments, where the UEs and the RRHs are illustrated by dots and rectangles, respectively. The non-CoMP UEs in $\S$ are marked red. The CoMP UEs in $\S$ are marked green. The CoMP UEs in $\J\backslash\S$ are marked blue. The other UEs are in light gray.
Basically, there are two ways for enhancing the capacity performance of $\S$: Using CoMP for UEs in $\S$ or UEs in $\J\backslash\S$. That the former generates benefits is apparent, since the spectrum efficiency would be increased for $\S$, which is a direct effect of using CoMP. As for the latter, since CoMP raises the RRH resource efficiency for serving UEs, using CoMP for $\J\backslash\S$ costs less resource than the non-CoMP case. As a result, there would be more available resource for scaling up the demand of UEs in $\S$. Furthermore, the reduction of an RRH load results in lower interference to the other RRHs, leading to an indirect effect for performance enhancement. On the other hand, we remark that not every UE benefits from being served by CoMP. In \figurename~\ref{fig:10scaling}, those UEs in light gray do not fulfill Lemma~\ref{lma:adding} in CoMP selection. Experimentally, forcing them to use CoMP leads to virtually no capacity improvement or even worse performance, due to that those UEs may only have one RRH being in good channel condition to them. Such UEs would not benefit from CoMP.

The number of UEs participating in CoMP has a strong influence on Metric~\ref{item:metric-alpha}). This is analyzed based on three observations as follows. As the first observation, the trends of Metric~\ref{item:metric-alpha}) and Metric~\ref{item:metric-CoMP}) are very similar. Both metrics are influenced by the size of $\S$. For the second observation, by Metric~\ref{item:metric-CoMP}), when the size of $\S$ is smaller, more UEs tend to be involved in CoMP, which is the reason why Metric~\ref{item:metric-alpha}) gets better when $|\S|$ becomes smaller.
The third observation explains why more UEs would be involved in CoMP when $|\S|$ becomes smaller. Note that though $|\J_{\text{(CoMP)}}|$ increases with the decrease of $|\S|$, $|\J_{\text{(CoMP)}}\cap\S|$ however decreases with $|\S|$ (see Metric \ref{item:metric-CoMP}) and Metric \ref{item:metric-CoMP-SR})). It means that, with the decrease of $|\S|$, more UEs outside $\S$ and fewer UEs inside $\S$ would participate in CoMP. The former increases faster than the reduction of the latter. Recall that using CoMP for $\S$ and $\J\backslash\S$ leads to direct and indirect effects of benefits. We conclude that the two effects affect the performance of $\S$ to different extents with respect to the group size of $\S$. The reason is that, if we scale up demands for many UEs, due to the resulted intensive traffic, directly using CoMP to raise the spectrum efficiency for $\S$ is the most effective way for improving the performance of $\alpha$. One can see by Metric \ref{item:metric-CoMP-SR}) that the number of CoMP UEs in $\S$ increases with $|\S|$. Though the network still gains from the indirect effect of CoMP, the benefit is much more limited compared to the direct effect. On the other hand, when only a small proportion of UEs requires demand scaling (i.e. small $|\S|$), the gain from the direct effect is limited by the number of UEs in $\S$. In this case, it is more beneficial to focus on the other UEs (i.e. $\J\backslash\S$) of which the number is much higher than $|\S|$. As a consequence, the indirect effect of CoMP becomes dominating, i.e., reducing the RRH load with CoMP in order to alleviate the interference to $\S$ for increasing $\alpha$. Algorithm~\ref{alg:MaxD} indeed seeks for a resource configuration that maximumly leverages the two effects in order to achieve the highest improvement by CoMP.
This is coherent with the snapshots in \figurename~\ref{fig:10scaling}. Visually, the ratio between CoMP UEs inside $\S$ over those outside $\S$ varies apparently with respect to $|\S|$.


%

In \figurename~\ref{fig:comp}, the user-centric capacity performance (i.e. Metric \ref{item:metric-alpha})) is different from the network-wise capacity performance (i.e. Metric \ref{item:metric-delivery})).
Employing a user-centric strategy may help to deliver considerably more bits to those UEs to be scaled, compared to scaling up the demand for all UEs. If the group size is small, from the network's point of view, the capacity improvement is not as much as being achieved by scaling the demand for all UEs. 
Next, we observe that Metric~\ref{item:metric-delivery}) has a strong correlation with Metric~\ref{item:metric-CoMP-SR}). The more UEs in $\S$ are served by CoMP, the more the bits can be delivered network-wisely. 
As one can see that though $|\S\cap\J_{\text{(CoMP)}}|$ increases with respect to $\S$, $|\J_{\text{(CoMP)}}|-|\S\cap\J_{\text{(CoMP)}}|$ decreases (see this by combining Metric~\ref{item:metric-CoMP-SR}) with Metric~\ref{item:metric-CoMP})), meaning that the indirect effect of CoMP degenerates with $|\S|$ because the UEs outside $\S$ becomes fewer.

As a conclusion, user-centric demand scaling benefits from CoMP directly as well as indirectly. Namely, serving UEs inside $\S$ and outside $\S$ with CoMP both contribute.

\subsection{Performance with Respect to $\bar{\rho}$}
\label{subsec:sim-rho}

In \figurename~\ref{fig:rhobar}, we show Metrics \ref{item:metric-alpha})--\ref{item:metric-CoMP-SR}) as function of $\bar{\rho}$. As expected, the larger $\bar{\rho}$ is, the higher improvement can be achieved via CoMP. By both Metric~\ref{item:metric-alpha}) and Metric~\ref{item:metric-delivery}), when there is better availability of resource, the impact of the size of $S$ on the achievable performance is higher. 
By Metric~\ref{item:metric-delivery}), one can see that the increase of total delivered demand is almost linear in $|\S|$ (when $\bar{\rho}$ is $0.6$ or higher). 
\textcolor[rgb]{0,0,0}{Due to the limitation of resource availability when $\bar{\rho}=0.4$, there may not be enough resource in RRHs for CoMP cooperations, even if we target scaling up demand for as many UEs as there are. As a consequence, the total delivered demand may not increase even if CoMP is enabled. One can see from Metric~\ref{item:metric-CoMP}) that the number of UEs in CoMP is lower when $\bar{\rho}=0.4$ than the cases with $\bar{\rho}\geq 0.6$. Specifically, the number of UEs in $\S$ that participate in CoMP is significantly lower when $\bar{\rho}=0.4$, as shown by Metric~\ref{item:metric-CoMP-SR}). As another observation, one can see that Metric~\ref{item:metric-delivery}) and Metric~\ref{item:metric-CoMP-SR}) have the same trend under all settings of $\bar{\rho}$. 
It indicates that the direct effect of CoMP for UEs in $\S$ has a high
correlation to the total delivered bits demand.}

By combining Metric~\ref{item:metric-CoMP}) with Metric~\ref{item:metric-CoMP-SR}), one can see that when $|\S|$ is small (e.g. $|\S|=20$), though $|\J_{\text{(CoMP)}}|$ increases quickly with $\bar{\rho}$, $|\J_{\text{(CoMP)}}\cap\S|$ does slowly. It means that the number of CoMP UEs outside $\S$ increases quickly with the increase of the available resource. Hence, when there is more resource, it is more flexible for RRHs to cooperate in CoMP to serve $\J\backslash\S$ for RRH load reduction.  As a consequence, the optimization leads to more UEs outside $\S$ to be involved in CoMP.

In conclusion, 
the availability of resource has an influence on the number of CoMP UEs.  When $|\S|$ becomes smaller, the UEs in $\S$ benefit more from increasing $\bar{\rho}$ and more UEs outside $\S$ should be served by CoMP.

\subsection{\textcolor[rgb]{0,0,0}{Performance with Respect to Resource Consumption}}

\textcolor[rgb]{0,0,0}{We show the relationship between user demand scaling/satisfaction and time-frequency resource allocation in \figurename~\ref{fig:load}. Namely, we evaluate how much RRH load (i.e. amount of resource) is required to achieve more demand delivery by optimizing CoMP selection. Our test scenarios have 19 cells. The radius of each cell is $100$ m. Each cell is deployed with $30$ UEs and $3$ RRHs. The other settings follow \tablename~\ref{tab:sim}.}

\textcolor[rgb]{0,0,0}{
\figurename~\ref{fig:load} compares the load levels of non-CoMP and CoMP, in order to show how much load is needed for achieving a $10\%$--$15\%$ increase of demand delivery by Algorithm~\ref{fig:load}. By observation, there is very slight difference in the total RRH load levels between non-CoMP and CoMP, though the latter delivers significantly more demand than the former. In addition, we found that optimizing the CoMP selection indeed reduces the average RRH load consumption for the highly-loaded RRHs. Therefore, we conclude that the time-frequency resource efficiency can be considerably improved by optimizing CoMP selection.
}

\subsection{\textcolor[rgb]{0,0,0}{CoMP Selection Comparison}}

\textcolor[rgb]{0,0,0}{
In this subsection, we compare our proposed CoMP selection mechanism with another one that is proposed for user-centric scenarios \cite{User-centric-comp}. We tailor the CoMP selection method in \cite{User-centric-comp} to be applicable in our model. We use the formula below as the utility of each UE for CoMP selection, which is designed for elastic application scenarios as considered by the authors.}
\begin{equation}
\textcolor[rgb]{0,0,0}{u_j(C_j) = \frac{\ln(C_j+1)}{\ln(d_j+1)}}
\label{eq:comp-utility}
\end{equation}
\textcolor[rgb]{0,0,0}{In \eqref{eq:comp-utility}, recall that $C_j$ and $d_j$ represent the achievable bits rate and the bits demand of UE $j$, respectively. Therefore we have $0\leq u_j \leq 1$. Suppose $\bm{\kappa}$ is the current RRH-UE association, and $\bm{\kappa}'$ the a new association that differs with $\bm{\kappa}$ only in one RRH-UE pair $(i,j)$ by $\kappa_{ij}=0$ and $\kappa'_{ij}=1$. By keeping the cell loads $\bm{\rho}$ fixed, we define the utility increase by adding the CoMP link $(i,j)$ as $\Delta u_{ij}=u_j(C_j(\bm{\kappa}'))-u_j(C_j(\bm{\kappa}))$. The utility-based CoMP selection rule is: $\kappa_{ij}=1$ if and only if $\Delta u_{ij}>0$ ($i\in\I$, $j\in\J$). Once a new CoMP link is added, Theorem~5 is used to obtain the optimum of the remaining resource allocation problem. Therefore, the resource allocation is performed in the same manner for both our CoMP selection method and the utility-based one. The difference is that our method considers the interference influence caused by the dynamic change of RRH loads.
}

\textcolor[rgb]{0,0,0}{The comparison is shown in \figurename~\ref{fig:comp_compare}. We can see that our proposed method outperforms the utility-based CoMP selection mechanism. In general, the improvement of $\alpha$ by our proposed CoMP seleciton is $1.5$ times more than the utility-based CoMP selection. We remark that, compared to the utility-based CoMP selection, our proposed method takes into account the dynamic change of the cell load when a new CoMP link is added (see Section \ref{subsec:association}). Hence our proposed method outperforms the utility-based one, if cell load coupling is taken into account.
}

\pgfplotsset{compat=1.11,
    /pgfplots/ybar legend/.style={
    /pgfplots/legend image code/.code={%
       \draw[##1,/tikz/.cd,yshift=-0.25em]
        (0cm,0cm) rectangle (3pt,0.8em);},
   },
}
\begin{figure}[t]
\centering
\begin{tikzpicture}
\begin{groupplot}[
    group style={
        group name=my plots,
        group size=1 by 2,
        xlabels at=edge bottom,
        xticklabels at=edge bottom,
        vertical sep=6pt
    },
	label style = {font=\fontsize{9pt}{10pt}\selectfont},
	legend style = {font=\fontsize{6pt}{10pt}\selectfont},
    ybar,
    enlarge x limits=0.2,
    legend style={
        anchor=north,legend columns=-1
     },
    legend pos = north west,
    ylabel={Improvement of $\alpha$ (\%)},
    xlabel={Percentage of scaled UEs (\%)},
    symbolic x coords={20, 40, 60, 80},
    xtick=data,
	minor x tick num=4,
	minor y tick num=4,
	major tick length=0.15cm,
	minor tick length=0.075cm,
	tick style={semithick,color=black},
	height=0.42\linewidth,
	width=0.9\linewidth,
    xtick align=inside,
    ymin = 0,
    ]
\nextgroupplot[ylabel style={align=center}, ylabel={Total load}, bar width = 0.25cm, 
	ymax = 30,
    legend entries={Non-CoMP, CoMP},
    legend columns = 4,
]
\addplot  coordinates {(20, 16.9522) (40, 20.1307)  (60, 20.1779) (80, 20.4903) };
\addplot  coordinates {(20, 16.9622) (40, 20.1507)  (60, 20.1979) (80, 20.5103) };

\nextgroupplot[ylabel style={align=center}, ylabel={Average load \\ of 10 highest \\ loaded RRHs }, bar width = 0.25cm]
\addplot  coordinates {(20, 0.790219) (40, 0.873644)  (60, 0.888588) (80, 0.90278) };
\addplot  coordinates {(20, 0.77107) (40, 0.853269)  (60, 0.865983) (80, 0.880204) };

\end{groupplot}
\end{tikzpicture}
\caption{\textcolor[rgb]{0,0,0}{This figures shows the time-frequency resource consumption (i.e. RRH load) for delivering user demands. Though not shown by the figure, we remark that, with non-CoMP being the baseline, CoMP leads to $10\%$--$15\%$ percentage of improvement on $\alpha$ for all the data in this figure.}}
\label{fig:load}
\end{figure}

\pgfplotsset{compat=1.11,
        /pgfplots/ybar legend/.style={
        /pgfplots/legend image code/.code={%
        \draw[##1,/tikz/.cd,bar width=3pt,yshift=-0.2em,bar shift=0pt]
                plot coordinates {(0cm,0.8em)};},
},
}
\begin{figure}[ht!] 
\centering 
\begin{tikzpicture}
\begin{axis}[
	log ticks with fixed point,
	xlabel={Number of scaled UEs $|\S|$},
	ylabel={Improvement of $\alpha$ (\%) },
	label style = {font=\fontsize{9pt}{10pt}\selectfont},
    legend style={ font=\fontsize{7.8pt}{10pt}\selectfont,
      anchor=north,legend columns=2},
      legend pos = north east,
minor x tick num=4,
minor y tick num=4,
major tick length=0.15cm,
minor tick length=0.075cm,
tick style={semithick,color=black},
	height=0.667\linewidth,
	width=\linewidth,
        scaled y ticks=false,
]

\addplot  [smooth, mark= diamond, color=red] 
		   coordinates{(10, 27.824) (20, 21.2774) (40, 16.483) (60, 14.3662) (80, 13.2022) (100, 13.0022)};
		   
\addplot  [smooth, mark= diamond, color=blue] 
		   coordinates{(10, 11.1345) (20, 8.2325) (40, 6.5923) (60, 5.7486) (80, 5.2321) (100, 5.2020)};
		   
\legend{Proposed CoMP, Utility-based CoMP}
\end{axis}
\end{tikzpicture} 
\caption{Comparison of CoMP selection regarding the improvement of demand scaling.}
\label{fig:comp_compare}
\end{figure}

\subsection{Convergence}
\label{subsec:convergence}

\pgfplotsset{compat=1.11,
        /pgfplots/ybar legend/.style={
        /pgfplots/legend image code/.code={%
        \draw[##1,/tikz/.cd,bar width=3pt,yshift=-0.2em,bar shift=0pt]
                plot coordinates {(0cm,0.8em)};},
},
}
\begin{figure}
\centering 
\begin{tikzpicture}
\begin{axis}[
	ymode = log,
	xlabel={Iteration $k$ in Theorem~\ref{thm:compute_G}},
	ylabel={$\norm{[\alpha^{(k)},\bm{\mu}^{(k)}]-[\alpha^{(k-1)},\bm{\mu}^{(k-1)}]}$},
	label style = {font=\fontsize{9pt}{10pt}\selectfont},
	legend cell align={left},
	legend pos = south west,
	legend style = {font=\fontsize{7pt}{10pt}\selectfont},
	axis background/.style={fill=white},
minor x tick num=0,
minor y tick num=4,
major tick length=0.15cm,
minor tick length=0.075cm,
xtick = {0,5,...,50},
tick style={semithick,color=black},
	height=0.667\linewidth,
	width=0.95\linewidth,
	xmin=5,
	xmax=30,
]
	
\addplot [smooth,dashed,color=black] coordinates {
(1, 67.3342) (2, 17.6731) (3, 1.13818) (4, 9.01713) (5, 3.65742) 
(6, 2.30692) (7, 1.20581) (8, 0.680668) (9, 0.372068) (10, 
0.206043) (11, 0.113538) (12, 0.0626687) (13, 0.0345765) (14, 
0.0190767) (15, 0.0105264) (16, 0.00580768) (17, 0.00320455) (18, 
0.00176808) (19, 0.000975566) (20, 0.000538269) (21, 0.000296995) 
(22, 0.000163868) (23, 0.0000904154) (24, 0.0000498872) (25, 
0.0000275256) (26, 0.0000151874) (27, 8.37973e-6) (28, 
 4.62356e-6) (29, 2.55108e-6) (30, 1.40757e-6) (31, 
 7.76637e-7) (32, 4.28514e-7) (33, 2.36435e-7) (34, 
 1.30454e-7) (35, 7.1979e-8) (36, 3.97148e-8) (37, 
 2.19129e-8) (38, 1.20906e-8) (39, 6.67104e-9) (40, 
 3.68079e-9) (41, 2.03089e-9) (42, 1.12055e-9) (43, 
 6.18265e-10) (44, 3.41139e-10) (45, 1.88223e-10) (46, 
 1.03832e-10) (47, 5.72697e-11) (48, 3.15978e-11) (49, 
 1.74367e-11) (50, 9.62075e-12)
 };

\addplot [smooth, color=brown] coordinates {
(1, 67.3342) (2, 56.0002) (3, 1.17635) (4, 0.40385) (5, 
0.05385) (6, 0.0174358) (7, 0.00592489) (8, 0.00174543) (9, 
0.000665538) (10, 0.000214173) (11, 0.0000740503) (12, 
0.0000251338) (13, 8.51233e-6) (14, 2.90076e-6) (15, 
 9.84223e-7) (16, 3.34665e-7) (17, 1.1372e-7) (18, 
 3.86461e-8) (19, 1.31347e-8) (20, 4.46376e-9) (21, 
 1.51706e-9) (22, 5.15589e-10) (23, 1.7523e-10) (24, 
 5.95506e-11) (25, 2.02345e-11) (26, 6.87805e-12) (27, 
 2.34301e-12) (28, 7.97584e-13) (29, 2.70006e-13) (30, 
 9.23706e-14)
};
\addplot [smooth, color=blue, dash dot] coordinates {
(1, 67.3342) (2, 60.5657) (3, 0.972442) (4, 0.10915) (5, 
0.0407204) (6, 0.00817776) (7, 0.003162) (8, 0.000842966) (9, 
0.000286363) (10, 0.0000839589) (11, 0.0000270435) (12, 
 8.22147e-6) (13, 2.59309e-6) (14, 7.99402e-7) (15, 
 2.50043e-7) (16, 7.75039e-8) (17, 2.41619e-8) (18, 
 7.50528e-9) (19, 2.33667e-9) (20, 7.2644e-10) (21, 
 2.26047e-10) (22, 7.02984e-11) (23, 2.18696e-11) (24, 
 6.80167e-12) (25, 2.11564e-12) (26, 6.57252e-13) (27, 
 2.04281e-13) (28, 6.30607e-14) (30, 6.30607e-15)
 };

\addplot [smooth, color=red, dotted,thick] coordinates {
(1, 67.3342) (2, 62.5134) (3, 0.561339) (4, 0.0800516) (5, 
0.0279898) (6, 0.00725912) (7, 0.00241841) (8, 0.000719891) (9, 
0.000231059) (10, 0.0000714372) (11, 0.0000226418) (12, 
 7.08476e-6) (13, 2.23604e-6) (14, 7.02543e-7) (15, 
 2.21413e-7) (16, 6.96684e-8) (17, 2.19459e-8) (18, 
 6.90915e-9) (19, 2.17607e-9) (20, 6.85229e-10) (21, 
 2.15806e-10) (22, 6.79607e-11) (23, 2.14024e-11) (24, 
 6.7395e-12) (25, 2.12275e-12) (26, 6.68798e-13) (27, 
 2.12275e-13) (28, 6.92779e-14) (29, 2.30926e-14) (30, 
 8.88178e-15) 
};
\legend{$|\S|=10$, $|\S|=40$, $|\S|=70$, $|\S|=100$}
\end{axis}
\end{tikzpicture}
\caption{This figures shows the norm $\norm{\cdot}$ in function of iteration $k$ in Theorem~\ref{thm:compute_G}, under $|\S|=10,40,70,100$. }
\label{fig:convergence}
\end{figure}

We show by \textcolor[rgb]{0,0,0}{Figure~\ref{fig:convergence}} the convergence of the proposed method in~\eqref{eq:G_iter} and~\eqref{eq:G_iter2}, with $|\S|=10,40,70,100$. Initially we have $\alpha^{(0)}=1$ and $\bm{\mu}^{(0)}=\bm{0}$. Numerically, one can see that both $\alpha$ and $\bm{\mu}$ converge fast. \textcolor[rgb]{0,0,0}{By Figure~\ref{fig:convergence}, under the same convergence tolerance parameter, the larger $|\S|$, the faster the convergence. Moreover, larger $|\S|$ requires fewer algorithmic iterations, as it becomes fast for RRHs to reach the resource limit (i.e. the condition $H(\bm{\rho})=1$). As for our convergence tolerance setting $\epsilon=10^{-4}$, the method converges within 25 iterations for all cases.} We remark that when $|\S|=100$, the iterations in Theorem~\ref{thm:compute_G} are actually computing the eigenvalue and eigenvector of a concave function and converge faster than the other cases.  In general, Theorem~\ref{thm:compute_G} serves as an efficient sub-routine for Algorithm~\ref{alg:MaxD}.

\subsection{Additional Results with Presence of Fronthaul Capacity}

In some scenarios, the capacity of the fronthaul links may turn out to
be the performance bottleneck.  Denote by the fronthaul link capacity
limit of RRH $i$ by $c_i$. Considering this capacity leads to the
additional constraint of $\sum_{j \in \J_i \cap \S} \alpha d_j \leq
c_i, i \in \R$.  

Algorithm~\ref{alg:MaxD} can be easily extended to incorporate
fronthaul capacity, by imposing a condition for
Step~\ref{step:add}. Namely, adding a link is considered, only if the
capacity limit of the fronthaul of the RRH in question is not reached
yet for the current demand scaling factor.

\begin{figure}[ht!]
\centering
\begin{tikzpicture}
\begin{axis}[
	log ticks with fixed point,
	xlabel={Normalized fronthaul capacity},
	ylabel={Demand scaling $\alpha$},
	label style = {font=\fontsize{9pt}{10pt}\selectfont},
    legend style={ font=\fontsize{6.8pt}{10pt}\selectfont,
      anchor=north,legend columns=1},
      legend pos = north west,
minor x tick num=4,
minor y tick num=4,
major tick length=0.15cm,
minor tick length=0.075cm,
tick style={semithick,color=black},
	height=0.667\linewidth,
	width=\linewidth,
        scaled y ticks=false,
]

\addplot  [smooth, mark= diamond, color=red] 
		   coordinates{(0.1, 12.5481) (0.2, 19.1293) (0.3, 24.3858) (0.4, 27.5176) (0.5, 28.8256) (0.8, 29.3156) (1.0, 29.4515)};
\addplot  [smooth, mark= diamond, color=blue] 
		   coordinates{(0.1, 7.08245) (0.2, 10.713) (0.3, 12.8424) (0.4, 14.43) (0.5, 15.3163) (0.8, 15.8123) (1.0, 16.0201)};
\addplot  [smooth, mark= diamond, color=green] 
		   coordinates{(0.1, 5.26059) (0.2, 7.3773) (0.3, 8.86115) (0.4, 9.87316) (0.5, 10.4942) (0.8, 10.67584) (1.0, 10.77584)};
\addplot  [smooth, mark= diamond, color=orange] 
		   coordinates{(0.1, 4.18025) (0.2, 5.80734) (0.3, 6.86266) (0.4, 7.61746) (0.5, 8.08098) (0.8, 8.13268) (1.0, 8.15768)};
		   
\legend{$|\S|=20$, $|\S|=40$, $|\S|=60$, $|\S|=80$}
\end{axis}
\end{tikzpicture} 
\caption{The impact of fronthaul capacity on demand scaling.}
\label{fig:fronthaul}
\end{figure}

In Figure~\ref{fig:fronthaul}, we show the impact of fronthaul
capacity on demand scaling. The maximum fronthaul capacity is set to a sufficiently large value such that it wouldn't be a bottleneck of the system performance. The capacity shown in the x-axis is normalized with respect to this value. As can be seen, demand
scaling is fronthaul-limited when the capacity is low, as increasing
this capacity consistently improves the scaling factor. However, the
curve becomes eventually saturated, meaning that the bottleneck is now
due to radio access instead of the fronthaul. The results also show
out algorithm is useful for studying which part of the network
imposes the performance-limiting factor.

\section{Conclusion}
\label{sec:conclusion}

We have proved how CoMP selection and demand scaling can be jointly
optimized and demonstrated how CoMP improves the performance of user
capacity in user-centric C-RAN.  We have revealed that the users
involved in demand scaling benefit both directly and indirectly from
CoMP, by increasing the spectrum efficiency, and alleviating the
interference among RRHs, respectively. Furthermore, the two effects
contribute to the performance to different degrees with respect to the
number of users for demand scaling. Finally, the user-centric demand
scaling method proposed in this paper is not limited by the C-RAN
architecture, and can be applied to other interference models that
fall into the SIF framework.

One extension of the work is to consider beamforming by deploying
multiple antennas at each RRH. This would bring additional performance
gains, in addition to what the current work has focused on.  As long
as the strategy of setting the beamforming vector is given, our
analysis remains applicable as the received signal and interference
terms remain linear. On the other hand, joint optimization of
association and beamforming, with presence of load coupling between
the cells, leads to a new type of optimization problem for our
forthcoming work.

\appendix

\section*{\small Proof of Theorem~\ref{thm:compute_G}}

\begin{lemma}
For any given $\alpha>0$, denote $\bm{\mu}_{\alpha}=\lim_{k\rightarrow\infty}\bm{\mu}_{\alpha}^{(k)}$ where 
\begin{equation}
\bm{\mu}_{\alpha}^{(k+1)}=\frac{\vec{F}_{\alpha}(\bm{\mu}^{(k)})}{H\circ\vec{F}_{\alpha}(\bm{\mu}^{(k)})}.
\label{eq:rho_alpha}
\end{equation}
Then $H(\bm{\mu}_{\alpha})=1$ and $\vec{F}_{\alpha}(\bm{\mu}_{\alpha})=\lambda_{\alpha}\bm{\mu}_{\alpha}$ for unique $\lambda_{\alpha} > 0$.
\label{lma:rho_alpha}
\end{lemma}
\begin{proof}
The lemma follows from Theorem 1 in~\cite{Krause:2001wd}.
\end{proof}
\begin{lemma}
For any given $\vec{\mu}\geq\vec{0}$, define $P:\mathbb{R}_{++}\rightarrow\mathbb{R}_{++}$: 
\begin{equation}
P:\alpha\mapsto\frac{1}{H\circ \vec{F}_{\alpha}(\bm{\mu})}.
\end{equation}
Then $P(\alpha)$ is an SIF.
\label{lma:P}
\end{lemma}
\begin{proof}
Given $\vec{\mu}$, denote $\varphi_j=f_j(\vec{\mu})$.
Then
\[
P(\alpha)=\frac{1}{\frac{1}{\bar{\rho}}\max_{i\in\R}\sum_{j\in\J_i}\frac{\varphi_j}{\pi_j(\alpha)}}
\]
For any $i\in\R$, the term $\sum_{j\in\J_i}\dfrac{\varphi_j}{\pi_j(\alpha)}$ is convex in $\alpha$.
Thus, the term $\max_{i\in\R}\sum_{j\in\J_i}\dfrac{\varphi_j}{\pi_j(\alpha)}$ is convex in $\alpha$, and thus the function $P(\alpha)$ is concave in $\alpha$ (strictly concave if there is at least one $j$ such that $\pi(\alpha)=\alpha$). The concavity implies the scalability. In addition, $P(\alpha)$ is monotonic in $\alpha$. Hence the conclusion.
\end{proof}

\begin{lemma}
For any $\bm{\mu}\in\mathbb{R}_{+}^{n}$ and any $\lambda\in\mathbb{R}_{++}$, $\lambda\geq\lambda_{\alpha}$ if $\lambda\bm{\mu}\geq\vec{F}_{\alpha}(\bm{\mu})$.
\label{lma:lambda_alpha}
\end{lemma}
\begin{proof}
The lemma follows from Theorem 13 in~\cite{Krause:2001wd}	.
\end{proof}

\begin{lemma}
Define $T:\mathbb{R}_{++}\rightarrow\mathbb{R}_{++}$:
\begin{equation}
T: \alpha\mapsto\frac{1}{H\circ\vec{F}_{\alpha}(\bm{\mu}_{\alpha})}
\label{eq:T}
\end{equation}
where $\bm{\mu}_{\alpha}$ follows the definition in~Lemma~\ref{lma:rho_alpha}. 
With $\S\neq\phi$, $T(\alpha)$ is an SIF.
\label{lma:sif}
\end{lemma}
\begin{proof}
By Lemma~\ref{lma:rho_alpha}, we have $\vec{F}_{\alpha}(\bm{\mu}_{\alpha})=\lambda_{\alpha}\bm{\mu}_{\alpha}$. Thus, for function $T$, $T(\alpha)=1/H(\lambda_{\alpha}\bm{\mu}_{\alpha})$ holds. Since $H(\bm{\mu}_{\alpha})=1$, we have
\begin{equation}
T(\alpha)=1/H\circ\vec{F}_{\alpha}(\bm{\mu}_{\alpha})=1/\lambda_{\alpha}H(\bm{\mu}_{\alpha})=1/\lambda_{\alpha}
\end{equation} 

 We first prove the monotonicity. Suppose $\alpha'>\alpha$. We use $\lambda_{\alpha'}$ and $\bm{\mu}_{\alpha'}$ to represent respectively the eigenvalue and the eigenvector of $\vec{F}_{\alpha'}$ such that $\vec{F}_{\alpha'}(\bm{\mu}')=\lambda_{\alpha'}\bm{\mu}_{\alpha'}$ and $H(\bm{\mu}_{\alpha'})=1$. For vector $\bm{\mu}_{\alpha}$, one can easily verify that $\vec{F}_{\alpha'}(\bm{\mu}_{\alpha})\leq\vec{F}_{\alpha}(\bm{\mu}_{\alpha})$. Therefore $\lambda_{\alpha}\bm{\mu}_{\alpha}\geq\vec{F}_{\alpha'}(\bm{\mu}_{\alpha})$. By Lemma~\ref{lma:lambda_alpha}, $\lambda_{\alpha'}\leq\lambda_{\alpha}$. Hence $T(\alpha')\geq T(\alpha)$ and the monotonicity holds.

We then prove the scalability. Consider for any $\eta>1$ the function $\frac{1}{\eta}\vec{F}_{\alpha}(\bm{\mu})$, and denote by $\lambda'$ and $\bm{\mu}'$ respectively its eigenvalue and eigenvector, i.e., $\frac{1}{\eta}\vec{F}_{\alpha}(\bm{\mu}')=\lambda'\bm{\mu}'$. Denote by $\lambda_{\eta\alpha}$ and $\bm{\mu}_{\eta\alpha}$ respectively the eigenvalue and eigenvector of function $\vec{F}_{\eta\alpha}$, i.e., $\vec{F}_{\eta\alpha}(\bm{\mu}_{\eta\alpha})=\lambda_{\eta\alpha}\bm{\mu}_{\eta\alpha}$. For vector $\bm{\mu}_{\eta\alpha}$, one can verify that the following relation holds.
\begin{equation}
\lambda_{\eta\alpha}\bm{\mu}_{\eta\alpha}=\vec{F}_{\eta\alpha}(\bm{\mu}_{\eta\alpha})\geq\frac{1}{\eta}F_{\alpha}(\bm{\mu}_{\eta\alpha})
\end{equation}
Based on Lemma~\ref{lma:lambda_alpha}, $\lambda'\leq\lambda_{\eta\alpha}$ holds. Specifically, with $\S\neq\phi$, we have for at least one $i\in\R$ such that $\pi_i(\alpha)=1$. We conclude $\lambda'<\lambda_{\eta\alpha}$ due to the monotonicity of $\vec{F}_{\alpha}(\bm{\mu})$ in $\bm{\mu}$. 
In addition, by~\eqref{eq:rho_alpha}, we have the equation below.
\begin{equation}
\bm{\mu}'=\lim_{k\rightarrow\infty}\frac{\frac{1}{\eta}\vec{F}_{\alpha}(\bm{\mu}^{(k)})}{H\circ\frac{1}{\eta}\vec{F}^{k}_{\alpha}(\bm{\mu})}=\lim_{k\rightarrow\infty}\frac{\vec{F}_{\alpha}(\bm{\mu}^{(k)})}{H\circ\vec{F}_{\alpha}(\bm{\mu}^{(k)})}=\bm{\mu}_{\alpha}
\end{equation}
Also, we have 
\begin{equation}
\frac{1}{\eta}\vec{F}_{\alpha}(\bm{\mu}')=\lambda'\bm{\mu}'\Leftrightarrow\vec{F}_{\alpha}(\bm{\mu}')=\eta\lambda'\bm{\mu}'.
\end{equation}
Therefore, by Lemma~\ref{lma:rho_alpha}, $\eta\lambda'=\lambda_{\alpha}$, i.e. $1/\lambda'=\eta/\lambda_{\alpha}$.
Combined with $\lambda'<\lambda_{\eta\alpha}$, we have
\begin{equation}
T(\eta\alpha)=\frac{1}{\lambda_{\eta\alpha}}<\frac{1}{\lambda'}=\frac{\eta}{\lambda_{\alpha}}=\eta T(\alpha).
\end{equation}
Hence the conclusion.
\end{proof}

We then prove Theorem~\ref{thm:compute_G} as follows.
\begin{proof}
The proof is straightforward, based on Lemmas~\ref{lma:rho_alpha}, Lemma~\ref{lma:P},~and Lemma~\ref{lma:sif}. Denote by $\alpha_{\bm{\mu}}^{(0)},\alpha_{\bm{\mu}}^{(1)},\ldots,\alpha_{\bm{\mu}}$ the sequence generated by $P(\alpha)$, with any $\alpha_{\mu}^{(0)}>0$, for any given $\bm{\mu}\geq \vec{0}$. By Lemma~\ref{lma:P}, $\alpha_{\bm{\mu}}$ is unique for $\bm{\mu}$. Similarly, denote by $\bm{\mu}^{(0)}_{\alpha},\bm{\mu}^{(1)}_{\alpha},\ldots,\bm{\mu}_{\alpha}$ the sequence generated by~\eqref{eq:rho_alpha}, with any $\bm{\mu}^{(0)}_{\alpha}\geq \vec{0}$, for any given $\alpha>0$. By~Lemma~\ref{lma:rho_alpha}, $\alpha_{\bm{\mu}}$ is unique for $\alpha$, and at the convergence we have $H(\bm{\mu}_{\alpha})=1$.
According to Lemma~\ref{lma:sif}, $T(\alpha)$ is an SIF of $\alpha$, then the $\alpha^{*}$ satisfying $\alpha^{*}=T(\alpha^{*})$ is unique. By fixing one of $\alpha$ and $\bm{\mu}$ and compute the sequence for the other alternately, the process falls into the category of asynchronous fixed point iterations, of which the convergence is guaranteed \cite[Page~434]{bertsekas1989parallel}. At the convergence, we have some $\bm{\mu}^{*}\geq\vec{0}$ such that
\begin{equation}
\bm{\mu}^{*}=\frac{\vec{F}_{\alpha}^{*}(\bm{\mu}^{*})}{H\circ\vec{F}_{\alpha}^{*}(\bm{\mu}^{*})},\text{ with }H(\bm{\mu}^{*})=1.
\end{equation}
Hence the conclusion.
\end{proof}

\section*{\small Proof of Theorem~\ref{thm:opt}}
\begin{proof}
The proof is based on the fact that the solution obtained from Theorem~\ref{thm:compute_G} fulfills Lemma~\ref{lma:sufficiency}.
First, by Lemma~\ref{lma:rho_alpha}, in the iterations in Theorem~\ref{thm:compute_G} we have $\vec{F}_{\alpha^{*}}(\bm{\mu}^{*})=\lambda_{\alpha^{*}}\bm{\mu}^{*}$ for a unique $\lambda_{\alpha^{*}}>0$ such that $H(\bm{\mu}^{*})=1$. Also, $\alpha^{*}=1/H\circ\vec{F}(\bm{\mu}^{*})$. Then, by combining these two equalities we get:
\begin{equation}
\alpha^{*}=\frac{1}{H(\lambda_{\alpha^{*}}\bm{\mu}^{*})}=\frac{1}{\lambda_{\alpha^{*}}H(\bm{\mu}^{*})}.
\end{equation}
Since $H(\bm{\mu}^{*})=1$, we then have:
\begin{equation}
\lambda_{\alpha^{*}}=\frac{1}{\alpha^{*}}.
\end{equation}
Hence we obtain the following derivation:
\begin{multline}
\vec{F}_{\alpha^{*}}(\bm{\mu}^{*})=\frac{1}{\alpha^{*}}\bm{\mu}^{*}\Leftrightarrow\alpha^{*}\vec{F}_{\alpha^{*}}(\bm{\mu}^{*})=\bm{\mu}^{*}  \Leftrightarrow \\
\frac{\alpha^{*}f_j(\bm{\mu})}{\pi_j(\alpha^{*})}=\mu_j^{*}~~j\in\J \Leftrightarrow \left\{
\begin{array}{l}
\mu_j=\alpha^{*}f_j(\bm{\mu}^{*})~~j\in\S \\
\mu_j=f_j(\bm{\mu}^{*})~~j\in\J\backslash\S
\end{array}\right.
\label{eq:equalities}
\end{multline}

By Lemma~\ref{lma:sufficiency}, i.e., the sufficient condition of optimality, the theorem hence holds.
\end{proof}

\bibliographystyle{IEEEtran}
\bibliography{ref}

\end{document}